\documentclass[a4paper,11pt]{article}%
\usepackage{amsfonts}
\usepackage{amsmath}
\usepackage{geometry}
\usepackage{hyperref}
\usepackage{amssymb}
\usepackage{graphicx}
\usepackage{caption}%
\providecommand{\U}[1]{\protect\rule{.1in}{.1in}}
\newtheorem{theorem}{Theorem}[section]
\newtheorem{acknowledgement}{Acknowledgement}

\newtheorem{claim}[theorem]{Assumption}

\newtheorem{corollary}{Corollary}[section]

\newtheorem{definition}{Definition}[section]

\newtheorem{lemma}{Lemma}[section]

\newtheorem{proposition}{Proposition}[section]

\newenvironment{proof}[1][Proof]{\noindent\textbf{#1.} }{\ \rule{0.5em}{0.5em}}
\begin{document}

\title{Estimation of anthracnose dynamics by nonlinear filtering}
\author{David Jaures FOTSA MBOGNE$^{a,}$\thanks{Corresponding author's address: email:
mjdavidfotsa@gmail.com, P.O. Box 455, ENSAI}\\$^{a}${\small Department of Mathematics and Computer Science, ENSAI, The
University of Ngaoundere}}
\date{}
\maketitle

\begin{abstract}
In this paper, we apply the nonlinear filtering theory to the estimation of
the partially observed dynamics of anthracnose which is a phytopathology. The
signal here is the inhibition rate and the observations are the fruit volume
ant the rotted volume. We propose stochastic models based on the deterministic
models given in the references \cite{fotsa,fotsa3}, in order to represent the
noise introduced by uncontrolled variation on parameters and errors on the
measurements. Under the assumption of Brownian noises we prove the well-posedness
the models either they take into account the space variable or not. The
filtering problem is solved for the non-spatial model giving Zakai and
Kushner-Stratonovich equations satisfied respectively by the unnormalized and
the normalized conditional distribution of the signal with respect to the
observations. A prevision problem and a discrete filtering problem are also studied for the realistic cases of discrete and possibly incomplete observations. We illustrate the filter behaviour through numerical simulations corresponding to different scenarios\\
\textbf{KeyWords--- }Anthracnose modelling, State estimation, Nonlinear filtering.\\
\textbf{AMS Classification--- }60H15, 60H10, 93E11, 93E10.

\end{abstract}

\section{Introduction}

\qquad Anthracnose is a phytopathology which occurs on several commercial
tropical crops . Among them the coffee is concerned by the coffee berry
disease (CBD) caused by the \textit{Colletotrichum} \textit{kahawae} which is
an ascomycete fungus
\cite{bieysse,boisson,chen,jeffries,mouen09,muller70,wharton}. In order to
understand, predict and control the disease dynamics, several models have been
proposed in the literature
\cite{danneberger,dodd,duthie,mouen07,mouen09,mouen072,mouen03,mouen08,wharton}%
. Recently, in \cite{fotsa} and \cite{fotsa2}, an evolution model with spatial
diffusion has been studied for anthracnose control. Optimal strategies were
computed with respect to given cost functionals. The general model surveyed in
\cite{fotsa} was given by the following equations:%

\begin{equation}
\partial_{t}\theta=\alpha\left(  t,x\right)  \left(  1-w\left(  t,x\right)
\theta\right)  +\operatorname{div}\left(  A\left(  t,x\right)  \nabla
\theta\right)  \text{, on }%
\mathbb{R}
_{+}^{\ast}\times U
\end{equation}%
\begin{equation}
\partial_{t}v=\frac{\beta\left(  t,x,\theta\right)  }{\eta\left(  t,x\right)
v_{\max}}\left(  \eta\left(  t,x\right)  v_{\max}-\frac{v}{1-\theta}\right)
\end{equation}%
\begin{equation}
\partial_{t}v_{r}=\frac{\gamma\left(  t,x,\theta\right)  }{v}\left(
v-v_{r}\right)
\end{equation}%
\begin{equation}
\theta\left(  0,x\right)  \in\left[  0,1\right[  ,\text{ }x\in\overline
{U}\subseteq%
\mathbb{R}
^{3}%
\end{equation}%
\begin{equation}
\left(  v\left(  0,x\right)  ,v_{r}\left(  0,x\right)  \right)  \in\left]
0,v_{\max}\right]  \times\left[  0,v_{\max}\right]  \text{, }x\in\overline
{U}\subseteq%
\mathbb{R}
^{3}%
\end{equation}
and%
\begin{equation}
\left\langle A\left(  t,x\right)  \nabla\theta\left(  t,x\right)  ,n\left(
x\right)  \right\rangle =0\text{, on }%
\mathbb{R}
_{+}^{\ast}\times\partial U
\end{equation}
where $n\left(  x\right)  $ denotes the normal vector on the boundary at $x$
and
\begin{equation}
w\left(  t,x\right)  =\frac{1}{1-\sigma u\left(  t,x\right)  }\text{.}%
\end{equation}

In the model above $\theta$ denotes the inhibition rate. The state variables
$v$ and $v_{r}$ are respectively the fruit volume density and the rot volume
density. The density $v$ is upper bounded by a value $v_{\max}$ which models
the natural fact that the fruit growth is limited. Nonnegative functions
$\alpha,\beta,\gamma$ characterize the effects of environmental and climatic
conditions on the rate of change of inhibition rate, fruit volume, and
infected fruit volume respectively \cite{danneberger,dodd,duthie}. There is a
control parameter $u$ representing the chemical strategy consisting on the
effects after application of fungicides. The parameter $1-\sigma\in\left]
0,1\right[  $ models the positive inhibition rate corresponding to epidermis
penetration by hyphae. Once the epidermis has been penetrated, the inhibition
rate cannot fall below this value, even under maximum control effort $(u=1)$.
Without any control effort $(u=0)$, the inhibition rate should increase
towards $1$. The environmental and climatic conditions affect the maximum
fruit volume through the $\left]  0,v_{\max}\right]  -$valued function $\eta$.
The term $\operatorname{div}\left(  A\nabla\theta\right)  $ refers to the
spatial spread of the disease in the open domain $U\subset\mathbb{R}^{3}$
which is assumed of class $C^{1}$. The boundary condition $\left\langle
A\nabla\theta,n\right\rangle =0$ where $A$ is a $3\times3$-matrix $\left(
a_{ij}\right)  $ could be understood as the law steering migration of the
disease between $U$ and its exterior. For instance, if $A$ is reduce to $I$
the identity matrix then $\left\langle \nabla\theta,n\right\rangle =0$ means
that the domain $U$ has no exchange with its exterior. The model in
\cite{fotsa2} has a similar form with the model given above. However the
authors added a new control strategy by impulses representing the harvesting
of pathogens with a given frequency.

In several cases, especially for the results in \cite{fotsa,fotsa2} on
anthracnose disease, optimal control strategies are given such as a feedbacks
and need to know the current state of the system and parameters values. It is
difficult in general to know exactly the trajectory of dynamic system.
Unfortunately \ the dynamics of the inhibition rate of anthracnose is not
exempt from that fact. However, it is more easier to observe volumes $v$ and
$v_{r}$. On the other hand, the global evolution of the system is subject to
pertubations coming from several sources. For instance, parameters of the
models vary depending on random climatic conditions and are often estimated
such as statistical averages. We can also mention errors occuring even during
every measurement of any output of the pathosystem. Those perturbations could be
taken into account through stochastic noises. The stochastic framework
presents several advantagies related to the use of large tools developped in
probability theory. As said before a probabilistic model enables to introduce
the randomness for some events. It also permits based on several observations
to smooth the model with better parameters. Another interest of stochastic
model is the possibility of estimation either of parameters or hidden states
using the displayed other states of the system. That last issue has been
widely studied in the framework of hidden Markov processes \cite{elliot}. The
corresponding attempts of solution in the large literature on the topic has
been regrouped on the name "filtering" \cite{bain,elliot,burkholderpardoux}.

The aim of this paper is to propose a stochastically noised model of
anthranose and to apply the filtering theory for the estimation of the
inhibition rate assuming that volumes $v$ and $v_{r}$ are observed. In the
remainder there is the following organization. In the Section
\ref{NotationEtDefinition}, we recall useful definitions adopt some notations that will be used later. The Section \ref{NoisedDynamicsModelling}
focuses on modelling and studying the well-posedness of the noised dynamics of
anthracnose either for the spatially distributed model or not. In the Section
\ref{StateEstimation}, we apply the filtering theory in order to determine the
law of the inhibition rate conditionally to the fruit volume and the rotted
volume. However, we first give in the Subsection \ref{NoisedDynamicsModelling2} an
equivalent model which is more suitable for the filtering procedure. Filtering
equations are derived into the Subsection
\ref{SubsectionContinuousObservationSDE}. The Section \ref{OtherIssues} is concerned by
resolution of a prevision problem in Subsection \ref{SubsectionPredictionSDE} and a
discrete filtering problem in Subsection \ref{SubsectionDiscreteFiltering}. We realize and discuss
several simulations in Section \ref{NumericalEstimations} in order to illustrate the behaviour of the filter for different scenarios. Finally, the paper ends with a
global discussion in Section \ref{DiscussionFiltering1}.

\section{Preliminaries}\label{NotationEtDefinition}

\qquad In this Subsection and the remaining of the paper we consider a
probability space $\left(  \Omega,\mathcal{F},P\right)  $ with a filtration
$\left(  \mathcal{F}_{t}\right)  _{t\geq0}$ such that $\mathcal{F}_{0}$
contains all negligible sets. Let $E$ denote a Banach space, $\mathcal{B}_{E}$
(or simply $\mathcal{B}$ when there is not ambiguity) the Borel $\sigma
$-algebra on $E$ and $\lambda_{E}$ the Lebesgue's measure. In order to
alleviate notations we will note $\left(  \Omega,\mathcal{F},P\right)  $ and
$\left(  E,\mathcal{B}_{E},\lambda_{E}\right)  $ simply by $\Omega$ and $E$.
When $E\subseteq%
\mathbb{R}
$ we simply note $\lambda_{E}$ by $\lambda$. If $F$ is another Banach space
then we note $\mathcal{L}\left(  E;F\right)  $ the space of linear continous
applications from $E$ to $F$, $\mathcal{L}\left(  E\right)  =\mathcal{L}%
\left(  E;E\right)  $, $E^{\prime}=\mathcal{L}\left(  E;%
\mathbb{R}
\right)  $ and $\mathcal{L}^{0}\left(  \Omega;E\right)  $ the set of
$E$-valued random variables\footnote{See \cite{pardoux} for the definition and
properties of random variables valued in Banach spaces.}.

\begin{definition}
Let $X\in\mathcal{L}^{0}\left(  \Omega;E\right)  $, $p\in\left]
0,\infty\right[  $ and $t\geq0$.

\begin{enumerate}
\item[$\left(  i\right)  $] $X\in\mathcal{L}^{p}\left(  \Omega;E\right)  $ if
$E\left[  \left\Vert X\right\Vert _{E}^{p}\right]  \equiv\int\nolimits_{\Omega
}\left\Vert X\left(  \omega\right)  \right\Vert _{E}^{p}dP\left(
\omega\right)  <\infty$.

\item[$\left(  ii\right)  $] $L^{p}\left(  \Omega;E\right)  $ is the set of
classes in $\mathcal{L}^{p}\left(  \Omega;E\right)  $ such that $\left[
X\right]  =\left[  Y\right]  $ if $E\left[  \left\Vert X-Y\right\Vert _{E}%
^{p}\right]  =0$.

\item[$\left(  iii\right)  $] $X\in L_{t}^{p}\left(  \Omega;E\right)  $ if
$X\in L^{p}\left(  \Omega;E\right)  $ and $X$ is $\mathcal{F}_{t}$-measurable.
\end{enumerate}
\end{definition}

\begin{definition}
Let $X\in\mathcal{L}^{0}\left(  \Omega;E\right)  $ and $t\geq0$.

\begin{enumerate}
\item[$\left(  i\right)  $] $X\in\mathcal{L}^{\infty}\left(  \Omega;E\right)
$ if there is a negligible set $\mathcal{N}\in\mathcal{F}$ and a positive
number $m$ such that $\forall\omega\in\Omega\setminus\mathcal{N}$, $\left\Vert
X\left(  \omega\right)  \right\Vert _{E}\leq m$.

\item[$\left(  ii\right)  $] $L^{\infty}\left(  \Omega;E\right)  $ is the set
of classes in $\mathcal{L}^{\infty}\left(  \Omega;E\right)  $ such that
$\left[  X\right]  =\left[  Y\right]  $ if there is a negligible set
$\mathcal{N}\in\mathcal{F}$ such that $\forall\omega\in\Omega\setminus
\mathcal{N}$, $\left\Vert X\left(  \omega\right)  -Y\left(  \omega\right)
\right\Vert _{E}=0$.

\item[$\left(  iii\right)  $] $X\in L_{t}^{\infty}\left(  \Omega;E\right)  $
if $X\in L^{\infty}\left(  \Omega;E\right)  $ and $X$ is $\mathcal{F}_{t}$-measurable.
\end{enumerate}
\end{definition}

\begin{definition}
Let $I\subseteq%
\mathbb{R}
_{+}$ and $X=\left(  X_{t}\right)  _{t\in I}$ such that $\forall t\in I,$
$X_{t}:\left(  \Omega,\mathcal{F},P\right)  \rightarrow\left(  E,\mathcal{B}%
_{E}\right)  $ is a random variable. Then $X$ is called a stochastic process.
$X$ is said $\left(  \mathcal{F}_{t}\right)  $-adapted if $\forall t\in I,$
$X_{t}$ is $\mathcal{F}_{t}$-measurable.
\end{definition}

\begin{definition}
A stochastic process $X$ is said progressively measurable if $\forall t\geq0,$
$X_{t}$ is $\mathcal{B}_{\left[  0,t\right]  }\otimes\mathcal{F}_{t}$-measurable.
\end{definition}

\begin{definition}
\label{BMDefinition}A process $\left(  B_{t}\right)  _{t\geq0}$ is called a
Brownian motion\footnote{See \cite{pardoux} for more details on the topic.
Also see \cite{curtain} page 134, for Hilbert valued Brownian motions.} on
$E^{\prime}$ the dual space of $E$ if the following conditions are satisfied.

\begin{enumerate}
\item[$\left(  i\right)  $] $\forall t\geq0$, $B_{t}$ is a linear form on
$E^{\prime}$.

\item[$\left(  ii\right)  $] $\forall x\in E^{\prime}$, the process $\left(
B_{t}\left(  x\right)  \right)  _{t\geq0}$ is a real Brownian
motion\footnote{See Chapter 1, Section 1.3 in \cite{pardouxRascanu}.}.

\item[$\left(  iii\right)  $] There is a self adjoint positive linear and
continuous operator $K:E^{\prime}\rightarrow E$ such that $\forall x,y\in
E^{\prime},$ $\forall s,t\geq0,$
\[
E\left[  \left(  B_{t}\left(  x\right)  -B_{s}\left(  x\right)  \right)
\left(  B_{t}\left(  y\right)  -B_{s}\left(  y\right)  \right)  \right]
=\left\langle x,Ky\right\rangle \left(  t-s\right)
\]
$K$ is called the associated covariance operator.
\end{enumerate}
\end{definition}

\begin{definition}
Let $V$ and $H$ be two hilbert spaces such that $V\subseteq H$ and $H$ is
identified with its dual space. Let consider $F:\Omega\times%
\mathbb{R}
_{+}\times V\rightarrow H$, $G:\Omega\times%
\mathbb{R}
_{+}\times V\rightarrow H\otimes E^{\prime}$, $\left(  B_{t}\right)  _{t\geq
0}$ a Brownian motion\footnote{See \cite{pardoux} for more details on the
topic. Also see \cite{curtain} page 143.} on $E^{\prime}$ and stochastic
differential equation:
\begin{equation}
\left\{
\begin{array}
[c]{l}%
dX_{t}=F\left(  t,X_{t}\right)  dt+\left\langle G\left(  t,X_{t}\right)
,dB_{t}\right\rangle \\
X_{0}=\xi
\end{array}
\right.  \text{.} \label{GeneralEDPS}%
\end{equation}
A progressively measurable process $X$ is called a (strong) solution of
$\left(  \ref{GeneralEDPS}\right)  $ on $\left[  0,T\right]  $ if it
satisfies
\begin{equation}
\int\nolimits_{0}^{t}\left\Vert F\left(  s,X_{s}\right)  \right\Vert
_{H}ds+\int\nolimits_{0}^{t}\left\Vert G\left(  s,X_{s}\right)  ^{\ast
}G\left(  s,X_{s}\right)  \right\Vert _{H}^{2}ds<\infty
\end{equation}
and
\begin{equation}
X_{t}=\xi+\int\nolimits_{0}^{t}F\left(  s,X_{s}\right)  ds+\int\nolimits_{0}%
^{t}\left\langle G\left(  s,X_{s}\right)  ,dB_{s}\right\rangle \text{.}%
\end{equation}

\end{definition}

If $E$ has a finite dimension $n$ and $\Phi:E\rightarrow%
\mathbb{R}
$ is a functional of class $C^{k}$ then we note $D^{l}\Phi$ the differential
of order $l=\left(  l_{1},\cdots,l_{n}\right)  \in%
\mathbb{N}
^{n}$ with $\sum\nolimits_{i=1}^{n}l_{i}\leq k$. When $n=1$ we simply note
$D\Phi$ instead of $D^{1}\Phi$.

\section{Modelling of the anthracnose noised dynamics}\label{NoisedDynamicsModelling}

\qquad In this section we construct stochastic (partial) differential equation
models which reflect the random behaviour of the anthracnose dynamics. As we
said before, that dynamics is subject to many random pertubations and measurements
on the system are also noised. We make the common choice to represent the
randomness of the system by Brownian motions. Indeed, the Brownian motion has
some good properties and they are several well-known results in the literature
concerning stochastic differential equations with Brownian noise. For instance
the Brownian motion has a continuous version and is a martingale. Those
properties are usefull for the regularity of the solution and the last one is
particularly useful for filtering. We formally note $\theta=\left(  \theta
_{t}\right)  _{t\geq0}$ the stochastic process such that $\forall\omega
\in\Omega$, $\theta_{t}\left(  \omega\right)  $ is a space dependent function
defined on $U\subseteq%
\mathbb{R}
^{3}$ and representing the spatial distribution of anthracnose inhibition
rate. In the same manner we note $\left(  v_{t}\right)  _{\geq0}$ and $\left(
v_{t}^{r}\right)  _{\geq0}$ the spatial processes of fruits volumes and rotted
volumes. We set $\rho_{t}v_{t}=v_{t}^{r}$.

\qquad We adopt the following model for every $\left(  t,x\right)  \in$ $%
\mathbb{R}
_{+}^{\ast}\times U$,
\begin{equation}
d\theta_{t}\left(  x\right)  =\left(  f_{1}\left(  t,x,\theta_{t}\left(
x\right)  \right)  +\pounds _{t}\theta_{t}\left(  x\right)  \right)
dt+g_{1}\left(  t,x,\theta_{t}\left(  x\right)  \right)  dB_{t}^{1}\left(
x\right)  \label{StochasticPDE1}%
\end{equation}%
\begin{equation}
dv_{t}\left(  x\right)  =f_{2}\left(  t,x,v_{t}\left(  x\right)  ,\theta
_{t}\left(  x\right)  \right)  dt+g_{2}\left(  t,x,v_{t}\left(  x\right)
\right)  dB_{t}^{2}\left(  x\right)  \label{StochasticPDE2}%
\end{equation}%
\begin{equation}
d\rho_{t}\left(  x\right)  =f_{3}\left(  t,x,\overline{v}_{t}\left(  x\right)
,\rho_{t}\left(  x\right)  ,\theta_{t}\left(  x\right)  \right)
dt+g_{3}\left(  t,x,\rho_{t}\left(  x\right)  \right)  dB_{t}^{3}\left(
x\right)  \label{StochasticPDE3}%
\end{equation}%
\begin{equation}
\theta_{0},\rho_{0},\frac{v_{0}}{v_{\max}}\in L_{0}^{\infty}\left(
\Omega;L^{\infty}\left(  U;\left[  0,1\right[  \right)  \right)
\label{StochasticPDE4}%
\end{equation}%
\begin{equation}
\left\langle A_{t}\left(  x\right)  \nabla\theta_{t},n\left(  x\right)
\right\rangle =0\text{, on }%
\mathbb{R}
_{+}^{\ast}\times\partial U \label{StochasticPDE5}%
\end{equation}
where $\forall\omega\in\Omega$, $\forall x\in U$, $\forall y=\left(
y_{1},y_{2},y_{3}\right)  \in%
\mathbb{R}
^{3}$, $\forall i\in\left\{  1,3\right\}  $,
\begin{equation}
g_{i}\left(  t,x,y_{i}\right)  =\delta_{i}\left(  t,x\right)  \kappa
_{i}\left(  y_{i}\right)
\end{equation}%
\begin{equation}
g_{2}\left(  t,x,y_{2}\right)  =\delta_{2}\left(  t,x\right)  \kappa
_{2}\left(  \frac{y_{2}}{v_{\max}}\right)
\end{equation}%
\begin{equation}
f_{1}\left(  t,x,y_{1}\right)  =\alpha\left(  t,x\right)  \left(
1-y_{1}w\left(  t,x\right)  \right)  ,
\end{equation}%
\begin{equation}
f_{3}\left(  t,x,y,z\right)  =\gamma\left(  t,x,y\right)  \left(
1-y_{3}\right)  ,
\end{equation}%
\begin{equation}
\pounds _{t}\theta_{t}\left(  \omega\right)  \left(  x\right)
=\operatorname{div}\left(  A\left(  t,x\right)  \nabla\theta_{t}\left(
\omega\right)  \left(  x\right)  \right)  ,
\end{equation}
and
\begin{equation}
f_{2}\left(  t,x,y\right)  =\frac{\beta\left(  t,x,y_{1}\right)  }{\eta\left(
t,x\right)  v_{\max}}\left(  \eta\left(  t,x\right)  v_{\max}-\frac{y_{2}%
}{1+\varepsilon-y_{1}}\right)  .
\end{equation}

The positive term $\varepsilon$ (very smaller than $1$) has been already
introduced in the reference \cite{fotsa3} and models the fact that even the
inhibition rate is near to is maximal value the volume of the fruit is remains
greater than a smallest value. We guess that lower bound value is in the
neighborhood of $\varepsilon v_{\max}\underset{t}{\min}\left\{  \eta\left(
t\right)  \right\}  $. On the other hand the term $\varepsilon$ permits to
avoid singularities in the model. In order to take into account the impacts of
random climatic changes in the model, the parameters are assumed to depend on the time. For $i$ belonging to $\left\{
1,2,3\right\}  $, $\left(  B_{t}^{i}\right)  _{t\geq0}$ is an $\left(
\mathcal{F}_{t}\right)  _{t\geq0}$-adapted cylindrical Brownian motion on the
dual of an Hilbert space to make precise later and its covariance operator is the
identity operator. The system of initial conditions both with the Brownian
motions is assumed independent. Each $\delta_{i}$ is positive function giving
the range of noises and $\kappa_{i}$ is a nonnegative locally lipschitz
continuous functions modelling the dependence of the noises with respect to
the state.

\subsection{A lumped model}\label{LumpedModel}

\qquad In this subsection we survey the model $\left(  \ref{StochasticPDE1}%
\right)  -\left(  \ref{StochasticPDE5}\right)  $ assuming that the diffusion
operator $\pounds _{t}$ is null. That can correspond to a situation where
disease spreading is limited either by natural climatic and relief conditions
or by a control strategy. Since there is not diffusion, the study is
restricted at each point and the space variable can be forgotten. Each
Brownian motion $\left(  B_{t}^{i}\right)  _{t\geq0}$ \ is assumed to be a
standard real Wiener process starting from zero. The model is then simpler to
study and however could display average behaviours and give an idea on the way
to study the general model. We then keep the same notations, omit the space
variable and remove the diffusion term in equation $\left(
\ref{StochasticPDE1}\right)  $. We assume that all the parameters of the model
are not random. The following assumptions are considered.

\begin{claim}
\label{HypothesisFilter1_1}$\forall i\in\left\{  1,2,3\right\}  $, $\delta
_{i}$, $\alpha\in L_{loc}^{\infty}\left(
\mathbb{R}
_{+};%
\mathbb{R}
_{+}\right)  $.
\end{claim}

\begin{claim}
\label{HypothesisFilter1_2}$u,\eta\in L^{\infty}\left(
\mathbb{R}
_{+};\left[  0,1\right]  \right)  $ and $\forall t\geq0$, $\inf\left\{
\eta\left(  s\right)  ;s\in\left[  0;t\right]  \right\}  >0$.
\end{claim}

\begin{claim}
\label{HypothesisFilter1_3}$\forall i\in\left\{  1,2,3\right\}  $, $\kappa
_{i}$ is a nonnegative locally Lipschitz continuous function which is positive
on the set $\left]  0,1\right[  $ and null on $%
\mathbb{R}
\setminus\left]  0,1\right[  $.
\end{claim}

\begin{claim}
\label{HypothesisFilter1_4}$\beta\in L_{loc}^{\infty}\left(
\mathbb{R}
_{+}\times%
\mathbb{R}
;%
\mathbb{R}
_{+}\right)  $ and $\gamma\in L_{loc}^{\infty}\left(
\mathbb{R}
_{+}\times%
\mathbb{R}
^{3};%
\mathbb{R}
_{+}\right)  $ satisfies
\end{claim}

\begin{enumerate}
\item[$\left(  i\right)  $] $\gamma$ is a measurable with respect to the two
first parameters and locally Lipschitz continuous with respect to the third parameter,

\item[$\left(  ii\right)  $] $\gamma\left(  t,.,.,.\right)  $ is increasing
with respect to the first parameter and $\gamma\left(  t,0,.,.\right)  $ is
decreasing with respect to the last parameter. Moreover, $\gamma\left(
t,.,.,0\right)  $ is nonnegative and such that
\begin{equation}
\gamma\left(  t,0,.,0\right)  =0=\gamma\left(  t,.,0,.\right)  \text{.}%
\end{equation}

\end{enumerate}

The assumption $\left(  \ref{HypothesisFilter1_4}\right)  -\left(  i\right)  $
guarantees that while the berry has a null volume (without berry) or the
disease has not started, the rot volume remains null. The assumption $\left(
\ref{HypothesisFilter1_4}\right)  -\left(  ii\right)  $ means that the rot
volume increases with the inhibition rate; when there is a not inhibition the
volume of rot does not increase while the fruit grows better and therefore the
proportion $\rho$ decreases. $\gamma$ could be chosen with the form
$\gamma\left(  t,y_{1},y_{2},y_{3}\right)  =\left(  \gamma_{1}\left(
t\right)  y_{1}-\gamma_{2}\left(  t\right)  y_{3}\right)  y_{2}$, with
$\gamma_{1},\gamma_{2}:%
\mathbb{R}
_{+}\times%
\mathbb{R}
^{3}\rightarrow%
\mathbb{R}
_{+\text{.}}$ Since all the coefficients of the simplified version of the
model $\left(  \ref{StochasticPDE1}\right)  -\left(  \ref{StochasticPDE5}%
\right)  $ are Lipschitz continuous with respect
to state variables we can apply Theorem 5.2.1 in \cite{oksendal} (page 66) to
conclude that there is unique solution (in the sense of indistinguishability)
defined on a maximal time set $\left[  0,T\right[  $ with $T\in%
\mathbb{R}
_{+}^{\ast}\cup\left\{  \infty\right\}  $. In the remainder of the subsection
we will establish that the solution is bounded in $\left[  0,1\right]  ^{3}$
and therefore that $T=\infty$.

\begin{lemma}
\label{BoundednessSDE}Let $\left(  \theta_{0},\frac{v_{0}}{v_{\max}},\rho
_{0}\right)  \in\left[  0,1\right]  ^{3}$, $P$-almost surely. If $\left(
\left(  \theta_{t},v_{t},\rho_{t}\right)  \right)  _{t\in\left[  0,T\right[
}$ is the solution of the lumped model $\left(  \ref{StochasticPDE1}\right)
-\left(  \ref{StochasticPDE5}\right)  $ then $P$-almost surely, $\forall
t\in\left[  0,T\right[  $, $\left(  \theta_{t},\frac{v_{t}}{v_{\max}},\rho
_{t}\right)  \in\left[  0,1\right]  ^{3}$.
\end{lemma}

\begin{proof}
Let $\varphi\in C^{2}\left(
\mathbb{R}
\right)  $ be a nonnegative function which is null on $\left[  0,1\right]  $
and positive elsewhere, decreases on $\left]  -\infty,0\right[  $ but
increases on $\left]  1,\infty\right[  $. An example of such a function
$\varphi$ is the map
\[
x\mapsto\left\{
\begin{array}
[c]{l}%
-x^{3}\text{, }x\leq0\\
0\text{, }0<x<1\\
\left(  x-1\right)  ^{3}\text{, }x\geq1
\end{array}
\right.
\]
Using the It\^{o} formula we have
\begin{align*}
d\varphi\left(  \theta_{t}\right)   &  =D\varphi\left(  \theta_{t}\right)
d\theta_{t}+\frac{1}{2}D^{2}\varphi\left(  \theta_{t}\right)  d\theta_{t}\cdot
d\theta_{t}\\
&  =f_{1}\left(  t,\theta_{t}\right)  D\varphi\left(  \theta_{t}\right)
dt+\frac{1}{2}D^{2}\varphi\left(  \theta_{t}\right)  \left(  g_{1}\left(
t,\theta_{t}\right)  \right)  ^{2}dt\\
&  +D\varphi\left(  \theta_{t}\right)  g_{1}\left(  t,\theta_{t}\right)
dB_{t}^{1}\\
&  =f_{1}\left(  t,\theta_{t}\right)  D\varphi\left(  \theta_{t}\right)  dt
\end{align*}
We can easily check that $f_{1}\left(  t,\theta_{t}\right)  D\varphi\left(
\theta_{t}\right)  $ is not positive and
\[
\frac{d\varphi\left(  \theta_{t}\right)  }{dt}\leq0
\]
The last inequality and the fact that $\varphi\left(  \theta_{0}\right)  =0$
imply that $\varphi\left(  \theta_{t}\right)  $ is not positive and therefore
null since $\varphi$ is a nonnegative function. Using the definition of
$\varphi$ we deduce that necessarily $\theta_{t}\in\left[  0,1\right]  $.
Using similar arguments and the fact that almost surely $\forall t\in\left]
0,T\right[  $, $\theta_{t}\in\left[  0,1\right]  $ we also obtain that
$\forall t\in\left]  0,T\right[  $, $\varphi\left(  \frac{v_{t}}{v_{\max}%
}\right)  =\varphi\left(  \rho_{t}\right)  =0$. Therefore $\left(  \rho
_{t},\frac{v_{t}}{v_{\max}}\right)  \in\left[  0,1\right]  ^{2}$.
\end{proof}

\begin{proposition}
\label{ExistenceAndUniquenessSDE}The lumped model has a unique (in the sense
of indistinguishability) solution defined on $%
\mathbb{R}
_{+}$.
\end{proposition}

\begin{proof}
Using Lemma \ref{BoundednessSDE} and the Theorem 5.2.1 in \cite{oksendal} the
result follows.
\end{proof}

\begin{lemma}
\label{RemainInteriorSDE}Let $\left(  \theta_{0},\frac{v_{0}}{v_{\max}}%
,\rho_{0}\right)  \in\left]  0,1\right[  ^{3}$, $P$-almost surely. If $\left(
\left(  \theta_{t},v_{t},\rho_{t}\right)  \right)  _{t\geq0}$ is the solution
of the lumped model then $P$-almost surely, $\forall t\geq0$, $\left(
\theta_{t},\frac{v_{t}}{v_{\max}},\rho_{t}\right)  \in\left]  0,1\right[
^{3}$.
\end{lemma}

Before giving a proof for the Lemma \ref{RemainInteriorSDE} we first recall a
particular version of the general comparison Proposition 3.12 in
\cite{pardouxRascanu} (page 149) :

\begin{proposition}
\label{CompareProposition}Let $f,\widetilde{f}:\Omega\times%
\mathbb{R}
_{+}\times%
\mathbb{R}
\rightarrow%
\mathbb{R}
$ and $g:\Omega\times%
\mathbb{R}
_{+}\times%
\mathbb{R}
\rightarrow%
\mathbb{R}
^{k}$ be three progressively measurable processes with respect to the first
two variables and continuous with respect to the third one. Let $W$ be a $k$
dimensional Brownian motion. Assume that $\forall t\geq0$, $P$-almost surely
the following inequalities hold :
\[
\int\nolimits_{0}^{t}\left\vert g\left(  s,X_{s}\right)  \right\vert
^{2}ds+\int\nolimits_{0}^{t}\left\vert g\left(  s,\widetilde{X}_{s}\right)
\right\vert ^{2}ds<\infty
\]
and
\[
\int\nolimits_{0}^{t}\left\vert f\left(  s,X_{s}\right)  \right\vert
ds+\int\nolimits_{0}^{t}\left\vert \widetilde{f}\left(  s,\widetilde{X}%
_{s}\right)  \right\vert ds<\infty
\]
where $X$ and $\widetilde{X}$ are solution of the following stochastic
differential equations :
\begin{equation}
X_{t}=X_{0}+\int\nolimits_{0}^{t}f\left(  s,X_{s}\right)  ds+\int
\nolimits_{0}^{t}\left\langle g\left(  s,X_{s}\right)  ,dW_{s}\right\rangle
\label{ComparisonSDE1}%
\end{equation}
and
\begin{equation}
\widetilde{X}_{t}=\widetilde{X}_{0}+\int\nolimits_{0}^{t}\widetilde{f}\left(
s,\widetilde{X}_{s}\right)  ds+\int\nolimits_{0}^{t}\left\langle g\left(
s,\widetilde{X}_{s}\right)  ,dW_{s}\right\rangle \text{.}
\label{ComparisonSDE2}%
\end{equation}
Also assume that there are two progressively measurable processes
$L,\ell:\Omega\times%
\mathbb{R}
_{+}\rightarrow%
\mathbb{R}
_{+}$ such that almost surely $\forall t\geq0$,
\[
\max\left\{  \int\nolimits_{0}^{t}\ell_{s}^{2}ds,\int\nolimits_{0}^{t}%
L_{s}ds\right\}  <\infty
\]
and $d\lambda\otimes dP$-almost everywhere, $\forall x,y\in%
\mathbb{R}
$,
\[
\left\vert f\left(  t,x\right)  -f\left(  t,x\right)  \right\vert \leq
L_{t}\left\vert x-y\right\vert
\]
and%
\[
\left\vert g\left(  t,X_{t}\right)  -g\left(  s,\widetilde{X}_{t}\right)
\right\vert \leq\ell_{t}\left\vert X_{t}-\widetilde{X}_{t}\right\vert \text{.}%
\]
Finally, assume that $P$-almost surely $X_{0}\geq\widetilde{X}_{0}$ and
$d\lambda\otimes dP$-almost everywhere on $\Omega\times%
\mathbb{R}
_{+}$, $f\left(  t,x\right)  \geq\widetilde{f}\left(  t,x\right)  $. Then

\begin{enumerate}
\item[$\left(  i\right)  $] $P$-almost surely, $\forall t\geq0$, $X_{t}%
\geq\widetilde{X}_{t}$ and $X$ is the unique solution of $\left(
\ref{ComparisonSDE1}\right)  $.

\item[$\left(  ii\right)  $] If moreover there exist $A\in\mathcal{F}$ and a
stopping time $\tau>0$ such that $\forall\omega\in A$, $X_{0}\left(
\omega\right)  >\widetilde{X}_{0}\left(  \omega\right)  $ or
\begin{equation}
\int\nolimits_{0}^{\tau\left(  \omega\right)  }\left(  f\left(  \omega
,s,X_{s}\right)  -\widetilde{f}\left(  \omega,s,\widetilde{X}_{s}\right)
\right)  ds>0
\end{equation}
then $\forall\omega\in A$, $X_{t}\left(  \omega\right)  >\widetilde{X}%
_{t}\left(  \omega\right)  $; in particular if $\forall\omega\in A$,
$X_{0}\left(  \omega\right)  >\widetilde{X}_{0}\left(  \omega\right)  $ then
$\forall\left(  \omega,t\right)  \in A\times%
\mathbb{R}
_{+}$, $X_{t}\left(  \omega\right)  >\widetilde{X}_{t}\left(  \omega\right)  $.
\end{enumerate}
\end{proposition}

\begin{proof}
(of the Lemma \ref{RemainInteriorSDE})

Let $\left(  \theta_{t}^{1}\right)  _{t\geq0},$ $\left(  \theta_{t}%
^{2}\right)  _{t\geq0},$ $\left(  \theta_{t}^{3}\right)  _{t\geq0}$ be three
solutions of the equation $\left(  \ref{StochasticPDE1}\right)  $ with the
respective initial conditions $0$, $\theta_{0}^{2}\in\left]  0,1\right[  $ and
$1$. Since $\alpha,w,\delta\in L_{loc}^{\infty}\left(
\mathbb{R}
_{+};%
\mathbb{R}
_{+}\right)  $ and $\kappa_{1}$ is Lipschitz continuous we use Lemma
\ref{BoundednessSDE} and apply Proposition \ref{CompareProposition} with
$f=\widetilde{f}=f_{1}$ and $g=g_{1}.$ Then $P$-almost surely $\forall t\geq
0$,
\[
0\leq\theta_{t}^{1}<\theta_{t}^{2}<\theta_{t}^{3}\leq1
\]
That gives the result for $\theta$. With the same manner we also etablish it
for $v$ and $\rho$; we just have to take respectively $f=\widetilde{f}%
=f_{i}\left(  .,\theta\right)  $ and $g=g_{i}$ for $i\in\left\{  1,2\right\}
$.
\end{proof}

Now let us give a stronger result useful for the filtering. For that we need
an additional assumption:

\begin{claim}
\label{HypothesisFilter1_5}There is a time $T^{\ast}>0$ such that there is a
nonempty interval $I\subset\left]  0,T^{\ast}\right[  $ satisfying $\forall
t\in I$, $\alpha\left(  t\right)  ,\beta\left(  t,.\right)  ,\gamma\left(
t,.\right)  >0$ .
\end{claim}

The Assumption \ref{HypothesisFilter1_5}\textbf{ }seems restrictive but is
still realistic because we are interested by the disease dynamics since
favourable conditions are fulfilled even for a small time.

\begin{lemma}
\label{LeaveEdgeSDE}Let at least one of the initial conditions be $P$-almost
surely null. Under the Assumption \ref{HypothesisFilter1_5}, if $\left(
\left(  \theta_{t},v_{t},\rho_{t}\right)  \right)  _{t\geq0}$ is the solution
of the lumped model then $P$-almost surely $\forall t>0$, $\left(  \theta
_{t},\frac{v_{t}}{v_{\max}},\rho_{t}\right)  \in\left]  0,1\right[  ^{3}$.
\end{lemma}

\begin{proof}
Following Lemmas \ref{BoundednessSDE} and \ref{RemainInteriorSDE} it is
sufficient to establish that if $P$-almost surely $\theta_{0}$, $\frac{v_{0}%
}{v_{\max}}$, $\rho_{0}$ are all null then $P$-almost surely $\forall
t\in\left]  0,T\right[  $, $\left(  \theta_{t},\rho_{t},\frac{v_{t}}{v_{\max}%
}\right)  \in\left]  0,1\right]  ^{3}$. We first give the proof that if
$P$-almost surely $\theta_{0}$ is null then $P$-almost surely, $\forall
t\in\left]  0,T\right[  $, $\theta_{t}\in\left]  0,1\right]  $. Let consider
the stopping time
\begin{equation}
\tau_{1}=\inf\left\{  t\geq0;\theta_{t}>0\text{ }\right\}
\end{equation}
Since $\theta$ is continuous, $\forall t\geq0$, $\theta_{t\wedge\tau_{1}}$ and
$E\left[  \int\nolimits_{0}^{t\wedge\tau_{1}}g_{1}\left(  s,\theta_{s}\right)
dB_{s}^{1}\right]  $ are all null and necessarily
\[
E\left[  \int\nolimits_{0}^{t\wedge\tau_{1}}f_{1}\left(  s,\theta_{s}\right)
ds\right]  =0
\]
Let $T^{\ast}$ be the random time given the Assumption
\ref{HypothesisFilter1_5}\textbf{.}
\begin{align*}
E\left[  \int\nolimits_{0}^{T^{\ast}}\alpha_{s}\mathbf{1}_{\left\{  s\leq
\tau_{1}\right\}  }ds\right]   &  =E\left[  \int\nolimits_{0}^{T^{\ast}%
\wedge\tau_{1}}\alpha\left(  s\right)  ds\right] \\
&  =E\left[  \int\nolimits_{0}^{T^{\ast}\wedge\tau_{1}}\alpha\left(  s\right)
\left(  1-w\left(  s\right)  \theta_{s}\right)  ds\right] \\
&  =E\left[  \int\nolimits_{0}^{T^{\ast}\wedge\tau_{1}}f_{1}\left(
s,\theta_{s}\right)  ds\right] \\
&  =0
\end{align*}
By the Assumption \ref{HypothesisFilter1_5}\textbf{, }$\forall t\in I$,
$\alpha\left(  t\right)  $ is positive and necessarily $\mathbf{1}_{\left\{
t\leq\tau_{1}\right\}  }$ is null. Hence, $\tau_{1}=0$. In the same manner let
consider the stopping times
\begin{equation}
\tau_{2}=\inf\left\{  t\geq0;v_{t}>0\right\}
\end{equation}
and
\begin{equation}
\tau_{3}=\inf\left\{  t\geq0;\rho_{t}>0\right\}  \text{.}%
\end{equation}
Using the Assumption \ref{HypothesisFilter1_5} and the fact that $\forall t\in
I$, $\theta_{t}\in\left]  0,1\right[  $ we have
\[
E\left[  \int\nolimits_{0}^{T^{\ast}\wedge\tau_{2}}f_{2}\left(  s,0,\theta
_{s}\right)  ds\right]  =0=E\left[  \int\nolimits_{0}^{T^{\ast}\wedge\tau_{3}%
}f_{3}\left(  s,v_{s},0,\theta_{s}\right)  ds\right]  .
\]
and $\beta\left(  t,\theta_{t}\right)  $ and $\gamma_{t}\left(  t,\theta
_{t},v_{t},\rho_{t}\right)  $ are positive. Necessarily $\tau_{i}=0$, $\forall
i\in\left\{  2,3\right\}  $ and the result follows.
\end{proof}

It seems important to mention that the Lemma \ref{BoundednessSDE}, the Lemma \ref{RemainInteriorSDE} and the Lemma \ref{LeaveEdgeSDE} remain true even if $g_{i}$ is indentically null.

\subsection{The distributed parameters model}\label{DistributedParamModel}

\qquad This subsection is devoted to the study of the full model $\left(
\ref{StochasticPDE1}\right)  -\left(  \ref{StochasticPDE5}\right)  $ with the
diffusion term. We expect to generalize results obtained for the lumped model.
Each process $\left(  B_{t}^{i}\right)  _{t\geq0}$ \ is now assumed to satisfy
the general defintion \ref{BMDefinition} on the dual space of the Sobolev
space $H^{2}\left(  U\right)  $ and is identified to an $H^{2}\left(
U\right)  $-valued process. As said before, we set identity operator as the
common covariance operator of each Brownian motion. We do not identify the
Hilbert space $H^{2}\left(  U\right)  $ with its dual space while we identify
the Lebesgue space $L^{2}\left(  U\right)  $ with its dual space. By the
Maurin Theorem\footnote{See \cite{clark}.} the embedding of $H^{2}\left(
U\right)  $ in $L^{2}\left(  U\right)  $ is of Hilbert-Schmidt type and
therefore the restriction on $L^{2}\left(  U\right)  $ of each Brownian motion
has a nuclear\footnote{We refer to \cite{gelfand} for the properties of
nuclear and Hilbert-Schmidt operators.} covariance operator. From the results
in \cite{gelfand} (Theorems 5-8 in chapter 1) and the Friedrichs Theorem
(Theorem 9.2 and Corollary 9.8 in \cite{brezis}) those covariance operators
have kernels as bilinear forms on $L^{2}\left(  U\right)  $. The following assumptions are considered.

\begin{claim}
\label{HypothesisFilter1_6}$\forall i\in\left\{  1,2,3\right\}  $, $\delta
_{i}$, $\alpha\in L_{loc}^{\infty}\left(
\mathbb{R}
_{+}\times\overline{U};%
\mathbb{R}
_{+}\right)  $.
\end{claim}

\begin{claim}
\label{HypothesisFilter1_7}$u,\eta\in L^{\infty}\left(
\mathbb{R}
_{+}\times\overline{U};\left[  0,1\right]  \right)  $ and $\forall t\geq0$,
$\inf\left\{  \eta\left(  s,x\right)  ;s\in\left[  0,t\right]  ,x\in
\overline{U}\right\}  >0$.
\end{claim}

\begin{claim}
\label{HypothesisFilter1_8}$\forall i\in\left\{  1,2,3\right\}  $, $\kappa
_{i}$ is a nonnegative locally Lipschitz continuous function which is positive
on the set $\left]  0,1\right[  $ and null on $%
\mathbb{R}
\setminus\left]  0,1\right[  $.
\end{claim}

\begin{claim}
\label{HypothesisFilter1_9}$\beta\in L_{loc}^{\infty}\left(
\mathbb{R}
_{+}\times\overline{U}\times%
\mathbb{R}
;%
\mathbb{R}
_{+}\right)  $ and $\gamma\in L_{loc}^{\infty}\left(
\mathbb{R}
_{+}\times\overline{U}\times%
\mathbb{R}
^{3};%
\mathbb{R}
_{+}\right)  $ satisfies
\end{claim}

\begin{enumerate}
\item[$\left(  i\right)  $] $\gamma$ is a measurable with respect to the two
first parameters and locally Lipschitz continuous with respect to the third parameter,

\item[$\left(  ii\right)  $] $\gamma\left(  t,x,.,.,.\right)  $ is increasing
with respect to the first parameter and $\gamma\left(  t,x,0,.,.\right)  $ is
decreasing with respect to the last parameter. Moreover, $\gamma\left(
t,x,.,.,0\right)  $ is nonnegative and such that
\begin{equation}
\gamma\left(  t,x,0,.,0\right)  =0=\gamma\left(  t,x,.,0,.\right)  \text{.}%
\end{equation}

\end{enumerate}

similarly to the lumped model, $\gamma$ could be chosen such as $\gamma\left(
t,x,y_{1},y_{2},y_{3}\right)  =\left(  \gamma_{1}\left(  t,x\right)
y_{1}-\gamma_{2}\left(  t,x\right)  y_{3}\right)  y_{2}$, with $\gamma
_{1},\gamma_{2}:%
\mathbb{R}
_{+}\times U\times%
\mathbb{R}
^{3}\rightarrow%
\mathbb{R}
_{+\text{.}}$

\begin{claim}
\label{HypothesisFilter1_10}$\forall i,j\in\left\{  1,2,3\right\}  $,
$a_{ij}\in L_{loc}^{\infty}\left(
\mathbb{R}
_{+}\times\overline{U};%
\mathbb{R}
\right)  $.
\end{claim}

\begin{claim}
\label{HypothesisFilter1_11}$\forall T\geq0$, $\exists C\left(  T\right)  \in%
\mathbb{R}
_{+}^{\ast}$ such that $\forall t\in\left[  0,T\right]  $, $\forall h\in%
\mathbb{R}
^{3}$, $\forall x\in U$,
\begin{equation}
\sum\nolimits_{i=1}^{n}\sum\nolimits_{j=1}^{n}a_{ij}\left(  t,x\right)
h_{i}h_{j}\geq C\sum\nolimits_{i=1}^{n}h_{i}^{2}\text{.}
\label{ConditionCoerciviteModelSPDE}%
\end{equation}

\end{claim}

Before giving the first result of this subsection we state a particular
version of the Theorem 2.1 (page 93) proved in \cite{pardoux} .

\begin{theorem}
\label{TheoremGeneralUniciteStochasticEvolutionPBFilter1}Let $\mathcal{X}$,
$\mathcal{Y}$ and $\mathcal{Z}$ denote three separable Hilbert spaces such
that $\mathcal{X}$ is continuously embedded and dense in $\mathcal{Y}$ which
is identified with its dual space. Let $\mathcal{X}^{\prime}$ and
$\mathcal{Z}^{\prime}$ denote respective dual spaces of $\mathcal{X}$ and
$\mathcal{Z}$, $p$ be a real number in $\left]  1,\infty\right[  $ and $T\in%
\mathbb{R}
_{+}$. Let also consider the stochastic differential equation given such as
$\phi_{0}\in L^{2}\left(  \Omega\times\left]  0,T\right[  ;\mathcal{Y}\right)
$ and $\forall t\in\left]  0,T\right[  $,
\begin{equation}
d\phi_{t}+F\left(  t,\phi_{t}\right)  dt+G\left(  t,\phi_{t}\right)
dB_{t}=f\left(  t,.\right)  dt \label{StochasticEvolutionPBFilter1}%
\end{equation}
where $\left(  B_{t}\right)  _{t\geq0}$ is a Brownian motion on $\mathcal{Z}%
^{\prime}$, $f\in L^{p^{\prime}}\left(  \Omega\times\left]  0,T\right[
;\mathcal{X}^{\prime}\right)  $ is non anticipative, the operators
$F:\mathcal{X\rightarrow X}^{\prime}$ and $G:\mathcal{Y\rightarrow L}\left(
\mathcal{Z};\mathcal{Y}\right)  $ are not necessarily linear but satisfy for
almost every $t\in\left]  0,T\right[  $ and independently on the choice of $t$
the following conditions:

\begin{enumerate}
\item[$\left(  i\right)  $] There are three real constants $c,\lambda$ and
$\nu$ such that $c>0$ and $\forall\psi\in\mathcal{X}$,
\[
2\left\langle F\left(  t,\psi\right)  ,\psi\right\rangle +\lambda\left\Vert
\psi\right\Vert _{\mathcal{Y}}^{2}+\nu\geq c\left\Vert \psi\right\Vert
_{\mathcal{X}}^{p}+\left\Vert G\left(  t,\psi\right)  \right\Vert
_{\mathcal{L}\left(  \mathcal{Z};\mathcal{Y}\right)  }^{2}\text{,}%
\]

\item[$\left(  ii\right)  $] $\forall\psi,\varphi\in\mathcal{X}$,
$\left\langle F\left(  t,\psi\right)  -F\left(  t,\varphi\right)
,\psi-\varphi\right\rangle +\lambda\left\Vert \psi-\varphi\right\Vert
_{\mathcal{Y}}^{2}\geq0$,

\item[$\left(  iii\right)  $] There is a constant $\mu\geq0$ such that
$\forall\psi\in\mathcal{X}$, $\left\Vert F\left(  t,\psi\right)  \right\Vert
_{\mathcal{X}^{\prime}}$ $\leq\mu\left\Vert \psi\right\Vert _{\mathcal{X}%
}^{p-1}$,

\item[$\left(  iv\right)  $] $\forall\psi,\varphi,\phi\in\mathcal{X}$, the
application $\xi\in%
\mathbb{R}
\mapsto\left\langle F\left(  t,\psi+\xi\varphi\right)  ,\phi\right\rangle $ is continuous,

\item[$\left(  v\right)  $] $\forall\psi\in\mathcal{X}$, $\forall\varphi
\in\mathcal{Y}$, the applications $t\in\left]  0,T\right[  \mapsto F\left(
t,\psi\right)  $ and $t\in\left]  0,T\right[  \mapsto G\left(  t,\varphi
\right)  $ are Lebesgue-measurable,

\item[$\left(  vi\right)  $] $G\left(  t,0\right)  =0$ and for every bounded
subset $S\subseteq\mathcal{Y}$, there is a constant $C\left(  S\right)  $ such
that $\forall\psi,\varphi\in S$, $\left\Vert G\left(  t,\psi\right)  -G\left(
t,\varphi\right)  \right\Vert _{\mathcal{L}\left(  \mathcal{Z};\mathcal{Y}%
\right)  }\leq C\left\Vert \psi-\varphi\right\Vert _{\mathcal{Y}}$.
\end{enumerate}

Then the equation $\left(  \ref{StochasticEvolutionPBFilter1}\right)  $ has a
unique adapted solution $\phi\in L^{p}\left(  \Omega\times\left]  0,T\right[
;\mathcal{X}\right)  \cap L^{2}\left(  \Omega;C\left(  \left]  0,T\right[
;\mathcal{Y}\right)  \right)  $.
\end{theorem}

\begin{proposition}
\label{ExistenceAndUniquenessSPDE}There is a process $\left(  \theta
_{t}\right)  _{t\geq0}$ valued in $H^{1}\left(  U\right)  $ which is the
unique solution of the equations $\left(  \ref{StochasticPDE1}\right)
,\left(  \ref{StochasticPDE4}\right)  $ and $\left(  \ref{StochasticPDE5}%
\right)  $.
\end{proposition}

\begin{proof}
We set $\mathcal{X=}H^{1}\left(  U;%
\mathbb{R}
\right)  $, $\mathcal{Y}=L^{2}\left(  U\right)  $, $\mathcal{Z}=H^{2}\left(
U;%
\mathbb{R}
\right)  $, $F\left(  t,\psi\right)  =\alpha\left(  t,.\right)  w\left(
t,.\right)  \psi-\pounds _{t}\left(  \psi\right)  $, $G\left(  t,\psi\right)
:\varphi\in H^{2}\left(  U;%
\mathbb{R}
\right)  \mathcal{\longmapsto}g_{1}\left(  t,\psi\right)  \varphi\in
L^{2}\left(  U\right)  $, $f\left(  t,.\right)  =\alpha\left(  t,.\right)  $
and $p=2$. Recall that $\alpha$, $f_{1}$ and $g_{1}$ are bounded. Moreover,
$f_{1}\left(  t,.\right)  $ and $g_{1}\left(  t,.\right)  $ are Lipschitz
continuous with respect to $\theta_{t}$ and $\kappa_{1}\left(  0\right)  =0$.
Hence conditions $\left(  ii\right)  -\left(  vi\right)  $ of the Theorem
\ref{TheoremGeneralUniciteStochasticEvolutionPBFilter1} are fulfilled. To get
the result it suffices to show that there are constants $c,\lambda$ and $\nu$
such that $\alpha>0$ and $\forall\psi\in H^{1}\left(  U;%
\mathbb{R}
\right)  $, $\forall t\geq0$,
\begin{equation}
-2\left\langle \pounds _{t}\left(  \psi\right)  ,\psi\right\rangle
_{L^{2}\left(  U\right)  }+\left(  2\alpha\left(  t,.\right)  w\left(
t,.\right)  +\lambda\right)  \left\Vert \psi\right\Vert _{L^{2}\left(
U\right)  }^{2}\geq c\left\Vert \psi\right\Vert _{H^{1}\left(  U;%
\mathbb{R}
\right)  }^{2}+\left\Vert G\left(  t,\psi\right)  \right\Vert _{\mathcal{L}%
\left(  H^{2}\left(  U;%
\mathbb{R}
\right)  ;L^{2}\left(  U\right)  \right)  }^{2}\text{.}%
\end{equation}
That condition is definitely satisfied using the inequality $\left(
\ref{ConditionCoerciviteModelSPDE}\right)  $, the boundedness of $\delta_{1}$
and the fact that $\kappa_{1}$ is locally Lipschitz continuous. Note that the
solution is just defined on a maximal set of time. However, we will show in
the sequel that the solution is bounded and therefore is defined for every time.
\end{proof}

\begin{corollary}
There is a process $\left(  \left(  \theta_{t},v_{t},\rho_{t}\right)  \right)
_{t\geq0}$ valued in $H^{1}\left(  U\right)  \times L^{2}\left(  U\right)
\times L^{2}\left(  U\right)  $ which is the unique solution of $\left(
\ref{StochasticPDE1}\right)  -\left(  \ref{StochasticPDE5}\right)  $.
\end{corollary}

\begin{proof}
Since the existence is proved for $\theta$ and there is not a particular space
differential operator in equations $\left(  \ref{StochasticPDE2}\right)
-\left(  \ref{StochasticPDE3}\right)  $ we can fix the space variable and
solve finite dimension stochastic differential equations. The proof then
consists just in application of the Theorem 5.2.1 in \cite{oksendal} (page
66). Using the continuity of the solution with respect to the initial
condition the process $\left(  \left(  \theta_{t},v_{t},\rho_{t}\right)
\right)  _{t\geq0}$ valued in $H^{1}\left(  U\right)  \times L^{2}\left(
U\right)  \times L^{2}\left(  U\right)  $.
\end{proof}

The regularity of the solution with respect to the space variable depends on
the regularity of model's parameters. In the remainder of the subsection, we
will establish that $\left(  \theta_{t}\left(x \right),\frac{v_{t}\left(x \right)}{v_{max}},\rho_{t}\left(x \right)\right)$ is valued in $\left[  0,1\right]  ^{3}$. From
now and in the rest of the paper, we consider an orthonormal complete basis
$\left\{  e_{j}\right\}  _{j\in%
\mathbb{N}
}$ of $L^{2}\left(  U\right)  $. For every $i\in\left\{  1,2,3\right\}  $,
each Brownian motion $B^{i}$ has the decomposition
\begin{align*}
B_{t}^{i}  &  =\sum\nolimits_{j\in%
\mathbb{N}
}\lambda_{j}B_{t}^{i,j}e_{j}\\
&  =\sum\nolimits_{j\in%
\mathbb{N}
}B_{t}^{i,j}\left(  \left(  Q\right)  ^{\frac{1}{2}}e_{j}\right)
\end{align*}
with $Q$ a positive nuclear operator, $\left\{  \lambda_{j}\right\}  _{j\in%
\mathbb{N}
}\subset%
\mathbb{R}
_{+}^{\ast}$, $Tr\left(  Q\right)  =\sum\lambda_{j}^{2}<\infty$ and $\left(
B_{t}^{i,j}\right)  _{t\geq0}$is a standard real Brownian motion. More
precisely, $Q=JJ^{\ast}$ where $J$ denotes the embedding of $H^{2}\left(
U\right)  $ into $L^{2}\left(  U\right)  $. Since $H^{1}\left(  U\right)  $ is
dense as a subspace of subset $L^{2}\left(  U\right)  $ we can choose
$\left\{  e_{j}\right\}  _{j\in%
\mathbb{N}
}\subset H^{1}\left(  U\right)  $. That choice is suitable for the rest of our developments.

\begin{lemma}
\label{BoundednessSPDE}Let $\theta_{0}$ be valued in $\left[  0,1\right]  $,
$P$-almost surely. If $\left(  \theta_{t}\right)  _{t\geq0}$ is the solution
of the equations $\left(  \ref{StochasticPDE1}\right)  ,\left(
\ref{StochasticPDE4}\right)  $ and$\left(  \ref{StochasticPDE5}\right)  $ then
$P$-almost surely $\forall t\geq0$, $\theta_{t}$ is valued in $\left[
0,1\right]  $.
\end{lemma}

\begin{proof}
Let $\varphi\in C^{2}\left(
\mathbb{R}
\right)  $ be a nonnegative function which is null on $\left[  0,1\right]  $
and positive elsewhere, decreases on $\left]  -\infty,0\right[  $ but
increases on $\left]  1,\infty\right[  $, $D^{2}\varphi$ is nonnegative, both
$D\varphi$ and $D^{2}\varphi$ are bounded. Let also condider the functional
\[
\Phi:h\in L^{2}\left(  U\right)  \longmapsto\int\nolimits_{U}\varphi\left(
h\left(  x\right)  \right)  dx\in%
\mathbb{R}
_{+}%
\]
We can take for instance the function $\varphi$ such as
\[
x\mapsto\left\{
\begin{array}
[c]{l}%
-6x-6,x\leq-2\\
x^{3}+6x^{2}+6x+2,-2<x\leq-1\\
-x^{3},-1<x\leq0\\
0,0<x\leq1\\
\left(  x-1\right)  ^{3},1<x\leq2\\
-x^{3}+9x^{2}-21x+15,2<x\leq3\\
6x-12,x>3
\end{array}
\right.
\]
By the Rellich-Kondrachov theorem\footnote{See Theorem 2.16 in \cite{brezis},
page 285.} since $U$ is assumed of class $C^{1}$, $H^{1}\left(  U\right)
\subset L^{p}\left(  U\right)  $ with completely continuous imbedding when $p$
is taken in the set $\left[  1,6\right[  $. Hence, if $h\in H^{1}\left(
U\right)  $ then $\varphi\circ h,D\varphi\circ h\in L^{2}\left(  U\right)  $
and therefore $\Phi$ and $D\Phi$ are well defined and continuous. Since
$D^{2}\varphi$ bounded the map $\Phi$ is twice differentiable and the map
$h\in L^{2}\left(  U\right)  \mapsto D^{2}\Phi\left(  h\right)  \in
\mathcal{L}\left(  L^{2}\left(  U\right)  \right)  $ is continuous when
$\mathcal{L}\left(  L^{2}\left(  U\right)  \right)  $ is endowed with its
weak*-topology. Moreover, $D\Phi\left(  H^{1}\left(  U\right)  \right)
\subset H^{1}\left(  U\right)  $ and $\forall g\in L^{2}\left(  U\right)  $
such that $\forall h\in H^{1}\left(  U\right)  $,
\begin{align*}
\int\nolimits_{U}\varphi\left(  h\left(  x\right)  \right)  g\left(  x\right)
dx  &  \leq c_{1}\left\Vert g\right\Vert _{L^{2}\left(  U\right)  }\left\Vert
h\right\Vert _{L^{2}\left(  U\right)  }\\
&  \leq c_{2}\left\Vert g\right\Vert _{L^{2}\left(  U\right)  }\left\Vert
h\right\Vert _{H^{1}\left(  U\right)  }%
\end{align*}%
\[
\left\Vert D\Phi\left(  h\right)  \right\Vert _{H^{1}\left(  U\right)  }\leq
C\left\Vert h\right\Vert _{H^{1}\left(  U\right)  }%
\]
with $c_{1}$ the Lipschitz constant of $\varphi$ and
\[
C=\sup\left\{  D^{2}\varphi\left(  x\right)  ;x\in%
\mathbb{R}
\right\}  \text{.}%
\]
Indeed, $\forall h\in H^{1}\left(  U\right)  $ the operators $D\Phi\left(
h\right)  $ and $D^{2}\Phi\left(  h\right)  $ have respectively the kernels
$D\varphi\circ h$ and $D^{2}\varphi\circ h$.

We can now apply the It\^{o} formula\footnote{See also \cite{daprato}%
\ (Theorem 7.21, page 133).} as in \cite{pardoux} (Theorem 4.2, page 65) and
obtain
\begin{align*}
\Phi\left(  \theta_{t}\right)   &  =\Phi\left(  \theta_{0}\right)
+\int\nolimits_{0}^{t}D\Phi\left(  \theta_{s}\right)  \left(  f_{1}\left(
s,.,\theta_{s}\right)  +\pounds _{s}\theta_{s}\right)  ds\\
&  +\int\nolimits_{0}^{t}D\Phi\left(  \theta_{s}\right)  G\left(  s,\theta
_{s}\right)  dB_{s}^{1}+\frac{1}{2}Tr\int\nolimits_{0}^{t}D^{2}\Phi\left(
\theta_{s}\right)  G\left(  s,\theta_{s}\right)  QG^{\ast}\left(  s,\theta
_{s}\right)  ds\\
&  =\Phi\left(  \theta_{0}\right)  +\int\nolimits_{0}^{t}\int\nolimits_{U}%
D\varphi\left(  \theta_{s}\left(  x\right)  \right)  \left(  f_{1}\left(
s,x,\theta_{s}\left(  x\right)  \right)  +\pounds _{s}\theta_{s}\left(
x\right)  \right)  dxds\\
&  +\sum\nolimits_{j=1}^{\infty}\int\nolimits_{0}^{t}\lambda_{j}%
\int\nolimits_{U}D\varphi\left(  \theta_{s}\left(  x\right)  \right)
g_{1}\left(  s,x,\theta_{s}\left(  x\right)  \right)  e_{j}\left(  x\right)
dxdB_{s}^{1,j}\\
&  +\frac{1}{2}\sum\nolimits_{j=1}^{\infty}\int\nolimits_{0}^{t}\lambda
_{j}^{2}\int\nolimits_{U}D^{2}\varphi\left(  \theta_{s}\left(  x\right)
\right)  g_{1}\left(  s,x,\theta_{s}\left(  x\right)  \right)  e_{j}%
^{2}\left(  x\right)  dx\\
&  \times\int\nolimits_{U}g_{1}\left(  s,y,\theta_{s}\left(  y\right)
\right)  e_{j}^{2}\left(  y\right)  dyds
\end{align*}
The integration with respect to the Brownian motion is guaranteed since
$\left(  g_{1}\left(  s,.,\theta_{s}\left(  .\right)  \right)  D\varphi\left(
\theta_{s}\left(  .\right)  \right)  \right)  _{s\geq0}$ is measurable and
bounded\footnote{See \cite{pardoux} and references therein.}. We can easily
check that $\Phi\left(  \theta_{0}\right)  =0$, $f_{t}^{1}\left(  \theta
_{t}\right)  D\varphi\left(  \theta_{t}\right)  $ is not positive and the last
two terms in the right side of the equality are null since $\kappa_{1}$ has
been assumed null on $%
\mathbb{R}
\setminus\left]  0,1\right[  $. Hence,%
\begin{align*}
\Phi\left(  \theta_{t}\right)   &  =\int\nolimits_{0}^{t}\int\nolimits_{U}%
\left(  f_{1}\left(  s,x,\theta_{s}\left(  x\right)  \right)  +\pounds _{s}%
\theta_{s}\left(  x\right)  \right)  D\varphi\left(  \theta_{s}\left(
x\right)  \right)  dxds\\
&  \leq\int\nolimits_{0}^{t}\int\nolimits_{U}D\varphi\left(  \theta_{s}\left(
x\right)  \right)  \pounds _{s}\theta_{s}\left(  x\right)  dxds\\
&  =-\int\nolimits_{0}^{t}\int\nolimits_{U}\left\langle A_{s}\left(  x\right)
\nabla\theta_{s},\nabla D\varphi\left(  \theta_{s}\right)  \right\rangle
dxds\\
&  =-\int\nolimits_{0}^{t}\int\nolimits_{U}\left\langle A_{s}\left(  x\right)
\nabla\theta_{s},D^{2}\varphi\left(  \theta_{s}\right)  \nabla\theta
_{s}\right\rangle dxds\\
&  \leq0
\end{align*}
The last inequality implies that $\Phi\left(  \theta_{t}\right)  $ is not
positive. Using the definition of $\Phi$ necessarily $\Phi\left(  \theta
_{t}\right)  $ is null and therefore $\theta_{t}\left(  x\right)  \in\left[
0,1\right]  $ for almost every $x\in U$.
\end{proof}

\begin{proposition}
Let $\left(  \frac{v_{0}}{v_{\max}},\rho_{0}\right)  $ be valued in $\left[
0,1\right]  ^{2}$, $P$-almost surely. If $\left(  \left(  \theta_{t}%
,v_{t},\rho_{t}\right)  \right)  _{t\in\left[  0,T\right[  }$ is the solution
of the model $\left(  \ref{StochasticPDE1}\right)  -\left(
\ref{StochasticPDE5}\right)  $ then $P$-almost surely $\forall t\geq0$,
$\left(  \frac{v_{t}}{v_{\max}},\rho_{t}\right)  $ is valued in $\left[
0,1\right]  ^{2}$.
\end{proposition}

\begin{proof}
Using the particular form of equations $\left(  \ref{StochasticPDE2}\right)
-\left(  \ref{StochasticPDE3}\right)  $ we can fix the space variable and
conclude as in Lemma \ref{BoundednessSDE}.
\end{proof}

It seems difficult to generalize the Lemma \ref{RemainInteriorSDE} and the Lemma \ref{LeaveEdgeSDE} in spatial case. However, we are able to give some results
similar to the Lemma \ref{RemainInteriorSDE}.

Let
\[
m=\inf\left\{  \min\left\{  \theta_{0}\left(  x\right)  ,1-\sigma u_{t}\left(
x\right)  \right\}  ;t\geq0,x\in U\right\}
\]
and
\[
2M=\inf\left\{  \min\left\{  1-\theta_{0}\left(  x\right)  ,\sigma
u_{t}\left(  x\right)  \right\}  ;t\geq0,x\in U\right\}  \text{.}%
\]

\begin{lemma}
\label{RemainInteriorPDETheta}Let $\theta_{0}$ be valued in $\left]
0,1\right[  $, $P$-almost surely. Let $\left(  \theta_{t}\right)  _{t\geq0}$
be the solution of the equations $\left(  \ref{StochasticPDE1}\right)
,\left(  \ref{StochasticPDE4}\right)  $ and $\left(  \ref{StochasticPDE5}%
\right)  $, then $P$-almost surely, the following statements hold

\begin{enumerate}
\item[$\left(  i\right)  $] If $m>0$ then $\forall t\leq0$, $\theta_{t}>0$.

\item[$\left(  ii\right)  $] If $M>0$ then $\forall t\geq0$, $\theta_{t}<1$.
\end{enumerate}
\end{lemma}

\begin{proof}
Let consider new processes $\theta^{1}=\theta-m$ and $\theta^{2}=\theta+M$. $\theta^{1}$ and $\theta^{2}$ satisfy respectively the equations
\begin{equation}
d\theta_{t}^{1}\left(  t,x\right)  =\alpha\left(  t,x\right)  \left(
1-mw_{t}\left(  x\right)  -w_{t}\left(  x\right)  \theta_{t}^{1}\left(
x\right)  \right)  dt+\pounds _{t}\theta_{t}^{1}\left(  x\right)
dt+g_{1}\left(  t,x,\theta_{t}^{1}\left(  x\right)  +m\right)  dB_{t}%
^{1}\left(  x\right)
\end{equation}
and
\begin{equation}
d\theta_{t}^{2}\left(  t,x\right)  =\alpha\left(  t,x\right)  \left(
1-w_{t}\left(  x\right)  \left(  \theta_{t}^{2}\left(  x\right)  -M\right)
\right)  dt+\pounds _{t}\theta_{t}^{2}\left(  x\right)  dt+g_{1}\left(
t,x,\theta_{t}^{2}\left(  x\right)  -M\right)  dB_{t}^{1}\left(  x\right)
\end{equation}
Remembering that $w\left(  t,x\right)  =\frac{1}{1-\sigma u\left(  t,x\right)
}$ and $\sigma\in\left]  0,1\right[  $ we can see that $1-mw\left(
t,x\right)  >0$ and $1-w\left(  t,x\right)  \left(  1-M\right)  <0$. Using
similar arguments to those in the proof of the Proposition
\ref{BoundednessSPDE} we still have that $\theta^{1}$ and $\theta^{2}$ are
valued in $\left[  0,1\right]  $. If $m>0$ then $\theta$ is necessarily valued
in $\left]  0,1\right]  $. In the same manner, if $M>0$ then $\theta$ is
necessarily valued in $\left[  0,1\right[  $.
\end{proof}

\begin{proposition}
\label{RemainInteriorPDEvRho}Let $\left(  \frac{v_{0}}{v_{\max}},\rho
_{0}\right)  $ be valued in $\left]  0,1\right[  ^{2}$, $P$-almost surely and
let $\left(  \left(  \theta_{t},v_{t},\rho_{t}\right)  \right)  _{t\geq0}$ be
the solution of the model $\left(  \ref{StochasticPDE1}\right)  -\left(
\ref{StochasticPDE5}\right)  $. If there is a random time $T^{\ast}>0$ and a
nonempty interval $I\subset\left]  0,T^{\ast}\right[  $ satisfying $\forall
t\in I$, $\forall x\in\overline{U}$, $\forall y\in%
\mathbb{R}
_{+}^{\ast}$, $\beta\left(  t,x,\theta_{t}\left(  x\right)  \right)  $ and
$\gamma\left(  t,x,\theta_{t}\left(  x\right)  ,v_{t}\left(  x\right)
,\rho_{t}\left(  x\right)  \right)  $ are positive, then $P$-almost surely,
$\forall t>0$, $\left(  \rho_{t},\frac{v_{t}}{v_{\max}}\right)  $ is valued in
$\left]  0,1\right[  ^{2}$.
\end{proposition}

\begin{proof}
Using Lemma \ref{RemainInteriorPDETheta} it suffices to focus on $v$ and $\rho$. As before we fix the space variable and refer to the Lemma
\ref{RemainInteriorSDE} for the rest of the proof.
\end{proof}

\section{State estimation of the lumped with partialobservations}\label{StateEstimation}

\qquad This section is concerned by the filtering problem which consists in
finding the conditional law of a signal with respect to an observation. In the
present case the signal is $\left(  \theta_{t}\right)  _{t\geq0}$ and
observation is the two dimensional process $\left(  \left(  v_{t},\rho
_{t}\right)  \right)  _{t\geq0}$. In the remainder of the section, $\left(
\mathcal{F}_{t}^{23}\right)  _{t\geq0}$ will denote the subfiltration
generated by the two dimensional process $\left(  \left(  v_{t},\rho
_{t}\right)  \right)  _{t\geq0}$ and all $P$-null sets. The aim of the
subsection \ref{NoisedDynamicsModelling2} is to give a general representation
of the noised dynamics of the process $\left(  \left(  v_{t},\rho_{t}\right)
\right)  _{t\geq0}$ which is appropriate for the usual filtering procedure.
Other sections are devoted to the state estimation only for the lumped model
taking into consideration several cases.

\subsection{Another modelling of the noised observations dynamics}\label{NoisedDynamicsModelling2}

\qquad In this subsection, we construct another stochastic (partial) differential
equation models for the noised dynamics of anthracnose. Although the previous
modelling given in Section \ref{NoisedDynamicsModelling} seems natural and
displays good properties, it is less practical for the estimation we aim to
carry out in the sequel of this work. Indeed, the terms multiplying the
Brownian motions can take the null value and that singularity makes difficult
the filtering procedure. On the other hand, the fact that $v$ and $\rho$ are
bounded is not suitable for the use of the Girsanov theorem which is key tool.
To deal with that issue we will use as usually in statistical modelling a
transformation of the interest variable. Since the state variables are bounded
valued, we can use a logistic transformation in order to obtain new variables
valued in the whole space $%
\mathbb{R}
$. Hence, let $\overline{v}$ and $\overline{\rho}$ satisfies the
'deterministic' parts of the equations $\left(  \ref{StochasticPDE2}\right)
-\left(  \ref{StochasticPDE3}\right)  $ with the initial conditions $\left(
\overline{v}_{0},\overline{\rho}_{0}\right)  =\left(  v_{0},\rho_{0}\right)
$; that is
\begin{equation}
d\overline{v}_{t}=f_{t}^{2}\left(  x,\overline{v}_{t}\left(  x\right)
,\theta_{t}\left(  x\right)  \right)  1|_{\left\{  v_{t}>0\right\}  }dt
\label{EquationMeanV}%
\end{equation}%
\begin{equation}
d\overline{\rho}_{t}=f_{t}^{3}\left(  x,\overline{v}_{t}\left(  x\right)
,\overline{\rho}_{t}\left(  x\right)  ,\theta_{t}\left(  x\right)  \right)
1|_{\left\{  \rho_{t}>0\right\}  }dt\text{.} \label{EquationMeanRho}%
\end{equation}
We can prove using arguments similar with the Section
\ref{NoisedDynamicsModelling} that if $\left(  \overline{v}_{0},\overline
{\rho}_{0}\right)  $ is valued in $\left[  0,v_{\max}\right]  \times\left[
0,1\right]  $ then the all process $\left(  \left(  \overline{v}_{t}%
,\overline{\rho}_{t}\right)  \right)  _{t\geq0}$ is also valued in $\left[
0,v_{\max}\right]  \times\left[  0,1\right]  $. Conditionally upon $\theta_{t}$,
the expectation of $\left(  v_{t},\rho_{t}\right)  $ is given by $\left(
\overline{v}_{t},\overline{\rho}_{t}\right)  $. The term $1|_{\left\{
v_{t}>0\right\}  }$ in the equation $\left(  \ref{EquationMeanV}\right)  $
ensures the realistic property that $\overline{v}$ remains null while $v$ is
null. The term $1|_{\left\{  \rho_{t}>0\right\}  }$ plays a similar role in
the equation $\left(  \ref{EquationMeanRho}\right)  $.

If we set
\begin{equation}
\overline{X}_{t}=\left\{
\begin{array}
[c]{l}%
\ln\left(  \frac{\overline{v}_{t}}{v_{\max}-\overline{v}_{t}}\right)  \text{,
if }0<\overline{v}_{t}<v_{\max}\\
0\text{, otherwise}%
\end{array}
\right.  \label{RelationMeanVX}%
\end{equation}
and
\begin{equation}
\overline{Y}_{t}=\left\{
\begin{array}
[c]{l}%
\ln\left(  \frac{\overline{\rho}_{t}}{1-\overline{\rho}_{t}}\right)  \text{,
if }0<\overline{\rho}<1\\
0\text{, otherwise}%
\end{array}
\right.  \label{RelationMeanRhoY}%
\end{equation}
then $\overline{X}$ and $\overline{Y}$ satisfies when $0<\overline{v}<v_{\max
}$ and $0<\overline{\rho}<1$ the following equations :
\begin{equation}
d\overline{X}_{t}\left(  x\right)  =\left(  \eta\left(  t,x\right)  \left(
1+\exp\left(  -\overline{X}_{t}\right)  \right)  -\frac{1}{1+\varepsilon
-\theta_{t}}\right)  \frac{\beta\left(  t,x,\theta_{t}\right)  \left(
1+\exp\left(  \overline{X}_{t}\right)  \right)  }{\eta\left(  t,x\right)
v_{\max}^{2}}dt
\end{equation}
and%
\begin{equation}
d\overline{Y}_{t}\left(  x\right)  =\left(  1+\exp\left(  -\overline{Y}%
_{t}\right)  \right)  \gamma\left(  t,x,\theta_{t},\frac{v_{\max}\exp\left(
\overline{X}_{t}\right)  }{1+\exp\left(  \overline{X}_{t}\right)  },\frac
{\exp\left(  \overline{Y}_{t}\right)  }{1+\exp\left(  \overline{Y}_{t}\right)
}\right)  dt\text{.}%
\end{equation}
A common additive introduction a Brownian noise in the dynamics of $\left(
\overline{X},\overline{Y}\right)  $ leads to a diffusion process $\left(
X,Y\right)  $ satisfying
\begin{equation}
dX_{t}\left(  x\right)  =d\overline{X}_{t}\left(  x\right)  +\delta_{2}\left(
t,x\right)  dB_{t}^{2}\left(  x\right)  \label{SPDE_TransformV}%
\end{equation}
and%
\begin{equation}
dY_{t}\left(  x\right)  =d\overline{Y}_{t}\left(  x\right)  +\delta_{3}\left(
t,x\right)  dB_{t}^{3}\left(  x\right)  \text{.} \label{SPDE_TransformRho}%
\end{equation}
The terms $\varepsilon$, $B_{t}^{2}$ and $B_{t}^{3}$ have the same definitions
given in the Section \ref{NoisedDynamicsModelling}.

Similarly to the relation between $\left(  \overline{v},\overline{\rho
}\right)  $ and $\left(  \overline{X},\overline{Y}\right)  $, we could assume
that when $v_{t}$ and $\rho_{t}$ are not null they satisfy respectively
\begin{equation}
v_{t}=\frac{v_{\max}\exp\left(  X_{t}\right)  }{1+\exp\left(  X_{t}\right)  }
\label{RelationVX}%
\end{equation}
and
\begin{equation}
\rho_{t}=\frac{\exp\left(  Y_{t}\right)  }{1+\exp\left(  Y_{t}\right)
}\text{.} \label{RelationRhoY}%
\end{equation}
Note that it is useless to start the filtering while $v=0$ since there is not
fruit. In the same order of idea, when $v\neq0$ and $\rho=0$ we can restrict
ourselves to the informations brought by the dynamics of $X$ . Indeed, when
$v$ and $\rho$ are respectively null the variation of $X$ and $Y$ are reduced
to a Brownian noise.

\subsection{State estimation with continuous observations}\label{SubsectionContinuousObservationSDE}

\qquad In this subsection, we assume that at each time $t\geq0$ all the
observations $\left(  \left(  v_{s},\rho_{s}\right)  \right)  _{t\geq s\geq0}$
are really known. However, instead of $\left(  \left(  v_{s},\rho_{s}\right)
\right)  _{t\geq s\geq0}$ we will use the equivalent process $\left(  \left(
X_{s},Y_{s}\right)  \right)  _{t\geq s\geq0}$ which satisfies the equations
\begin{equation}
dX_{t}=f\left(  t,\overline{X}_{t},\theta_{t}\right)  dt+\delta_{2}\left(
t\right)  dB_{t}^{2}%
\end{equation}
and
\begin{equation}
dY_{t}=g\left(  t,\overline{X}_{t},\overline{Y}_{t},\theta_{t}\right)
dt+\delta_{3}\left(  t\right)  dB_{t}^{3}\text{.}%
\end{equation}
where
\begin{equation}
f\left(  t,\overline{X}_{t},\theta_{t}\right)  =\left(  \eta\left(  t\right)
\left(  1+\exp\left(  -\overline{X}_{t}\right)  \right)  -\frac{1}%
{1+\varepsilon-\theta_{t}}\right)  \frac{\beta\left(  t,\theta_{t}\right)
\left(  1+\exp\left(  \overline{X}_{t}\right)  \right)  }{\eta\left(
t\right)  v_{\max}} \label{EnonceF}%
\end{equation}
and%
\begin{equation}
g\left(  t,\overline{X}_{t},\overline{Y}_{t},\theta_{t}\right)  =\left(
1+\exp\left(  -\overline{Y}_{t}\right)  \right)  \gamma\left(  t,\theta
_{t},\frac{v_{\max}\exp\left(  \overline{X}_{t}\right)  }{1+\exp\left(
\overline{X}_{t}\right)  },\frac{\exp\left(  \overline{Y}_{t}\right)  }%
{1+\exp\left(  \overline{Y}_{t}\right)  }\right)  \text{.} \label{EnonceG}%
\end{equation}
We make the following necessary assumption until the end of the section:

\begin{claim}
\label{HypothesisFilter1_12}$\forall i\in\left\{  1,2,3\right\}  ,$
$\delta_{i}\in L_{loc}^{\infty}\left(
\mathbb{R}
;%
\mathbb{R}
_{+}\right)  $ and $\forall t\geq0,$ $\inf\left\{  \delta_{i}\left(  s\right)
;s\in\left[  0,t\right]  \right\}  >0$.
\end{claim}

Let adopt $\forall t\geq0$, the formal definition
\begin{align}
Z_{t}  &  =\exp\left(  -\frac{1}{2}\int\nolimits_{0}^{t}\left(  \frac
{f^{2}\left(  s,\overline{X}_{s},\theta_{s}\right)  }{\delta_{2}^{2}\left(
s\right)  }+\frac{g^{2}\left(  s,\overline{X}_{s},\overline{Y}_{s},\theta
_{s}\right)  }{\delta_{3}^{2}\left(  s\right)  }\right)  ds\right) \\
&  \times\exp\left(  -\int\nolimits_{0}^{t}\left(  \frac{f\left(
s,\overline{X}_{s},\theta_{s}\right)  }{\delta_{2}\left(  s\right)  }%
dB_{s}^{2}+\frac{g\left(  s,\overline{X}_{s},\overline{Y}_{s},\theta
_{s}\right)  }{\delta_{3}\left(  s\right)  }dB_{s}^{3}\right)  \right)
\nonumber
\end{align}
The following lemma holds.

\begin{lemma}
\label{ProbabilityChangeSDE}If $\left(  \theta_{0},\frac{v_{0}}{v_{\max}}%
,\rho_{0}\right)  \in\left[  0,1\right]  ^{3}$ then under the probability $P$,
$Z$ is an $\left(  \mathcal{F}_{t}\right)  $-martingale. Moreover, $\forall
t\geq0$, $Z_{t}$ is the Radon-Nikodym derivative of the restriction of a
probability $\widetilde{P}$ on $\mathcal{F}_{t}^{23}$ with respect to the
restriction of $P$ on $\mathcal{F}_{t}^{23}$:%
\begin{equation}
Z_{t}=\frac{d\widetilde{P}%
\vert
_{\mathcal{F}_{t}^{23}}}{dP%
\vert
_{\mathcal{F}_{t}^{23}}}\text{.}%
\end{equation}

\end{lemma}

\begin{proof}
The process $\left(  Z_{t}\right)  _{t\geq0}$ is $\mathcal{F}_{t}^{23}%
$-adapted. Using the properties of the solution $\left(  \overline
{v},\overline{\rho}\right)  $ of the equations $\left(  \ref{EquationMeanV}%
\right)  $ and $\left(  \ref{EquationMeanRho}\right)  $, the relations
$\left(  \ref{RelationMeanVX}\right)  $ and $\left(  \ref{RelationRhoY}%
\right)  $, and the properties of functions $f$ and $g$ given by $\left(
\ref{EnonceF}\right)  $ and $\left(  \ref{EnonceG}\right)  $ the following
Novikov condition is satisfied for every $t\geq0$:
\begin{equation}
E\left[  \exp\left(  \frac{1}{2}\int\nolimits_{0}^{t}\left(  \frac
{f^{2}\left(  s,\overline{X}_{s},\theta_{s}\right)  }{\delta_{2}^{2}\left(
s\right)  }+\frac{g^{2}\left(  s,\overline{X}_{s},\overline{Y}_{s},\theta
_{s}\right)  }{\delta_{3}^{2}\left(  s\right)  }\right)  ds\right)  \right]
<\infty\text{.}%
\end{equation}
Therefore using Proposition 2.50 in \cite{pardouxRascanu} (page124) and the
It\^{o} formula to compute $dZ_{t}$, we deduce that $\left(  Z_{t}\right)  $
is an $\left(  \mathcal{F}_{t}^{23}\right)  $-martingale which satisfies
$E\left[  Z_{t}\right]  =1$ and there are probabilities $\widetilde{P}_{t}$
such that $\forall t\geq0$,%
\begin{equation}
Z_{t}^{n}=\frac{d\widetilde{P}_{t}}{dP%
\vert
_{\mathcal{F}_{t}^{23}}}\text{.}%
\end{equation}
Using the Daniell-Kolmogorov-Tulcea Theorem A.12 stated in \cite{bain} (page
302) there is a probability $\widetilde{P}$ on $\mathcal{F}$ such that its
restriction on $\mathcal{F}_{t}^{23}$ is $\widetilde{P}_{t}$.
\end{proof}

The Lemma \ref{ProbabilityChangeSDE} gives a change of probability which will
be very useful in the remainder of the subsection. If we set
\begin{equation}
\widetilde{B}_{t}^{2}=\int\nolimits_{0}^{t}\frac{dX_{s}}{\delta_{2}\left(
s\right)  }\text{ }%
\end{equation}
and
\begin{equation}
\widetilde{B}_{t}^{3}=\int\nolimits_{0}^{t}\frac{dY_{s}}{\delta_{3}\left(
s\right)  }%
\end{equation}
then $\left(  \left(  \widetilde{B}_{t}^{2},\widetilde{B}_{t}^{3}\right)
\right)  _{t\geq0}$ is a Brownian motion under the probability $\widetilde{P}%
$. Let
\begin{align}
\widetilde{Z}_{t}  &  =\exp\left(  -\frac{1}{2}\int\nolimits_{0}^{t}\left(
\frac{f^{2}\left(  s,\overline{X}_{s},\theta_{s}\right)  }{\delta_{2}%
^{2}\left(  s\right)  }+\frac{g^{2}\left(  s,\overline{X}_{s},\overline{Y}%
_{s},\theta_{s}\right)  }{\delta_{3}^{2}\left(  s\right)  }\right)  ds\right)
\\
&  \times\exp\left(  \int\nolimits_{0}^{t}\left(  \frac{f\left(
s,\overline{X}_{s},\theta_{s}\right)  }{\delta_{2}\left(  s\right)
}d\widetilde{B}_{s}^{2}+\frac{g\left(  s,\overline{X}_{s},\overline{Y}%
_{s},\theta_{s}\right)  }{\delta_{3}\left(  s\right)  }d\widetilde{B}_{s}%
^{3}\right)  \right)  \text{.}\nonumber
\end{align}
Under $\widetilde{P}$, $\widetilde{Z}_{t}$ has the same properties of $Z_{t}$
under $P$ and
\begin{equation}
\widetilde{Z}_{t}=\frac{dP%
\vert
_{\mathcal{F}_{t}^{23}}}{d\widetilde{P}%
\vert
_{\mathcal{F}_{t}^{23}}}\text{.}%
\end{equation}
Moreover, $\left(  E\left[  \widetilde{Z}_{t}%
\vert
\mathcal{F}_{\infty}^{23}\right]  \right)  _{t\geq0}$ is an $\left(
\mathcal{F}_{t}^{23}\right)  $-martingale under $\widetilde{P}$ and has a
continuous version (see Proposition 2.3.1 in \cite{burkholderpardoux}).

In the following, we set $\forall t\geq0$, $\pi_{t}\left(  \varphi\right)
\equiv E\left[  \varphi\left(  \theta_{t}\right)
\vert
\mathcal{F}_{t}^{23}\right]  $ where $\varphi$ is a measurable function such
that
\begin{equation}
E\left[  \left\vert \varphi\left(  \theta_{t}\right)  \right\vert \right]
=\widetilde{E}\left[  \widetilde{Z}_{t}\left\vert \varphi\left(  \theta
_{t}\right)  \right\vert \right]  <\infty. \label{bonPhiFilter1}%
\end{equation}
Now we recall the useful Proposition in \cite{bain,burkholderpardoux}.

\begin{proposition}
(Kallianpur-Striebel)

If $\varphi$ satisfies the condition $\left(  \ref{bonPhiFilter1}\right)  $
then $P$ and $\widetilde{P}$-almost surely the following equality holds:%
\begin{equation}
\pi_{t}\left(  \varphi\right)  =\frac{\widetilde{E}\left[  \widetilde{Z}%
_{t}\varphi\left(  \theta_{t}\right)
\vert
\mathcal{F}_{t}^{23}\right]  }{\widetilde{E}\left[  \widetilde{Z}_{t}%
\vert
\mathcal{F}_{t}^{23}\right]  }. \label{LoiFiltrageSDEFilter1}%
\end{equation}

\end{proposition}

There are instructive comments on a more general but simalar process $\left(
\pi_{t}\right)  _{t\geq0}$ in \cite{bain} especially in the Theorem 2.1 (page
14). For arbitrary $\varphi\in L^{\infty}\left(
\mathbb{R}
;%
\mathbb{R}
\right)  $ and $\forall t\geq0$, let
\begin{equation}
\zeta_{t}=\widetilde{E}\left[  \widetilde{Z}_{t}%
\vert
\mathcal{F}_{t}^{23}\right]
\end{equation}
and
\begin{equation}
\varsigma_{t}\left(  \varphi\right)  =\zeta_{t}\pi_{t}\left(  \varphi\right)
.
\end{equation}
Then the equality $\left(  \ref{LoiFiltrageSDEFilter1}\right)  $ becomes
\begin{equation}
\pi_{t}\left(  \varphi\right)  =\frac{\varsigma_{t}\left(  \varphi\right)
}{\varsigma_{t}\left(  \mathbf{1}\right)  }. \label{LoiFiltrageSDEvNormalise}%
\end{equation}
Let $\left(  \widetilde{\mathcal{F}}_{t}^{23}\right)  $ denote the filtration
generated by the two dimensional process $\left(  \left(  \widetilde{B}%
_{t}^{2},\widetilde{B}_{t}^{3}\right)  \right)  _{t\geq0}$. Naturally, we have
$\widetilde{\mathcal{F}}_{t}^{23}\subseteq\mathcal{F}_{t}^{23}$ and conversely
$\mathcal{F}_{t}^{23}\subseteq\widetilde{\mathcal{F}}_{t}^{23}$ holds since
the following equations has unique solutions:%
\begin{equation}
X_{t}=\int\nolimits_{0}^{t}\delta_{2}\left(  s\right)  d\widetilde{B}_{s}^{2}%
\end{equation}
and
\begin{equation}
Y_{t}=\int\nolimits_{0}^{t}\delta_{3}\left(  s\right)  d\widetilde{B}_{s}^{3}.%
\end{equation}
Since $\widetilde{\mathcal{F}}_{t}^{23}=\mathcal{F}_{t}^{23}$ and $\left(
\left(  \widetilde{B}_{t}^{2},\widetilde{B}_{t}^{3}\right)  \right)  _{t\geq
0}$ is a Brownan motion under $\widetilde{P}$ we have
\begin{equation}
\varsigma_{t}\left(  \varphi\right)  =\widetilde{E}\left[  \widetilde{Z}%
_{t}\varphi\left(  \theta_{t}\right)
\vert
\mathcal{F}_{\infty}^{23}\right].
\end{equation}
The \textit{unnormalized} law $\varsigma$ of $\theta$ is given by the following

\begin{theorem}
\label{TheoremZakaiEqSDE}(The Zakai equation)

If $O\subseteq%
\mathbb{R}
$ is an open set containing $\left[  0,1\right]  $ and $\varphi\in
C^{2}\left(  O\right)  $ then $\forall t\geq0$,
\begin{equation}
\varsigma_{t}\left(  \varphi\right)  =\varsigma_{0}\left(  \varphi\right)
+\int\nolimits_{0}^{t}\varsigma_{s}\left(  \mathcal{A}_{s}^{1}\varphi\right)
ds+\int\nolimits_{0}^{t}\varsigma_{s}\left(  \mathcal{A}_{s}^{2}%
\varphi\right)  dX_{s}+\int\nolimits_{0}^{t}\varsigma_{s}\left(
\mathcal{A}_{s}^{3}\varphi\right)  dY_{s} \label{ZakaiEqSDE}%
\end{equation}
where
\begin{equation}
\mathcal{A}_{t}^{1}\varphi\left(  x\right)  =f_{1}\left(  t,x\right)
\varphi^{\prime}\left(  x\right)  +\frac{1}{2}g_{1}^{2}\left(  t,x\right)
\varphi^{\prime\prime}\left(  x\right)  \text{,} \label{Operator1Filter1}%
\end{equation}%
\begin{equation}
\mathcal{A}_{t}^{2}\varphi\left(  x\right)  =\frac{f\left(  t,\overline{X}%
_{t},\theta_{t}\right)  }{\delta_{2}^{2}\left(  t\right)  }\varphi\left(
x\right),  \label{Operator2}%
\end{equation}
and%

\begin{equation}
\mathcal{A}_{t}^{3}\varphi\left(  x\right)  =\frac{g\left(  t,\overline{X}%
_{t},\overline{Y}_{t},\theta_{t}\right)  }{\delta_{3}^{2}\left(  t\right)
}\varphi\left(  x\right)  \text{.} \label{Operator3}%
\end{equation}

\end{theorem}

Before giving the proof of the Theorem \ref{TheoremZakaiEqSDE} we first state
an adapted version of the Lemma 2.2.4 proved in \cite{burkholderpardoux} (page 83).

\begin{lemma}
\label{LemmaPermutIntegralExpect} Let $\left(  \xi_{t}\right)  _{t\geq0}$ be
an $\mathcal{F}_{t}$-progressive process such that $\forall t\geq0$,
\[
E\left[  \int\nolimits_{0}^{t}\xi_{s}^{2}ds\right]  <\infty
\]
then
\[
\widetilde{E}\left[  \int\nolimits_{0}^{t}\xi_{s}dB_{s}^{1}%
\vert
\mathcal{F}_{\infty}^{23}\right]  =0,
\]%
\[
\widetilde{E}\left[  \int\nolimits_{0}^{t}\xi_{s}dX_{s}%
\vert
\mathcal{F}_{\infty}^{23}\right]  =\int\nolimits_{0}^{t}\widetilde{E}\left[
\xi_{s}%
\vert
\mathcal{F}_{\infty}^{23}\right]  dX_{s}%
\]
and%
\[
\widetilde{E}\left[  \int\nolimits_{0}^{t}\xi_{s}dY_{s}%
\vert
\mathcal{F}_{\infty}^{23}\right]  =\int\nolimits_{0}^{t}\widetilde{E}\left[
\xi_{s}%
\vert
\mathcal{F}_{\infty}^{23}\right]  dY_{s}\text{.}%
\]

\end{lemma}

\begin{proof}
(of Theorem \ref{TheoremZakaiEqSDE})

Let consider the probability $\widetilde{P}$.%

\begin{equation}
d\varphi\left(  \theta_{t}\right)  =\mathcal{A}_{t}^{1}\varphi\left(
\theta_{t}\right)  dt+g_{1}\left(  t,\theta_{t}\right)  \varphi^{\prime
}\left(  \theta_{t}\right)  dB_{t}^{1},%
\end{equation}%
\begin{equation}
d\widetilde{Z}_{t}=\mathcal{A}_{t}^{2}\widetilde{Z}_{t}dX_{t}+\mathcal{A}%
_{t}^{3}\widetilde{Z}_{t}dY_{t},%
\end{equation}
and
\begin{align}
\widetilde{Z}_{t}\varphi\left(  \theta_{t}\right)   &  =\widetilde{Z}%
_{0}\varphi\left(  \theta_{0}\right)  +\int\nolimits_{0}^{t}\widetilde{Z}%
_{s}\mathcal{A}_{s}^{1}\varphi ds+\int\nolimits_{0}^{t}\widetilde{Z}_{s}%
g_{1}\left(  s,\theta_{s}\right)  \varphi^{\prime}\left(  \theta_{s}\right)
dB_{s}^{1}\\
&  +\int\nolimits_{0}^{t}\widetilde{Z}_{s}\mathcal{A}_{s}^{2}\varphi
dX_{s}+\int\nolimits_{0}^{t}\widetilde{Z}_{s}\mathcal{A}_{s}^{3}\varphi
dY_{s}.\nonumber
\end{align}
We now use Lemma \ref{LemmaPermutIntegralExpect} since $\left(  \theta
,X,Y\right)  $ is continuous and bounded, $\varphi\in C^{2}\left(  O\right)  $
and the parameters of the model are locally bounded
with respect to the time and Lipschitz continuous with respect to the other
variables.
\begin{align*}
\varsigma_{t}\left(  \varphi\right)   &  =\widetilde{E}\left[  \widetilde
{Z}_{t}\varphi\left(  \theta_{t}\right)
\vert
\mathcal{F}_{\infty}^{23}\right] \\
&  =\widetilde{E}\left[  \widetilde{Z}_{0}\varphi\left(  \theta_{0}\right)
\vert
\mathcal{F}_{\infty}^{23}\right]  +\widetilde{E}\left[  \int\nolimits_{0}%
^{t}\widetilde{Z}_{s}\mathcal{A}_{s}^{1}\varphi\left(  \theta_{s}\right)  ds%
\vert
\mathcal{F}_{\infty}^{23}\right]  +\widetilde{E}\left[  \int\nolimits_{0}%
^{t}\widetilde{Z}_{s}g^{1}\left(  s,\theta_{s}\right)  \varphi^{\prime}\left(
\theta_{s}\right)  dB_{s}^{1}%
\vert
\mathcal{F}_{\infty}^{23}\right] \\
&  +\widetilde{E}\left[  \int\nolimits_{0}^{t}\widetilde{Z}_{s}\mathcal{A}%
_{s}^{2}\varphi\left(  \theta_{s}\right)  dX_{s}%
\vert
\mathcal{F}_{\infty}^{23}\right]  +\widetilde{E}\left[  \int\nolimits_{0}%
^{t}\widetilde{Z}_{s}\mathcal{A}_{s}^{3}\varphi\left(  \theta_{s}\right)
dY_{s}%
\vert
\mathcal{F}_{\infty}^{23}\right] \\
&  =\varsigma_{0}\left(  \varphi\right)  +\int\nolimits_{0}^{t}\widetilde
{E}\left[  \widetilde{Z}_{s}\mathcal{A}_{s}^{1}\varphi\left(  \theta
_{s}\right)
\vert
\mathcal{F}_{\infty}^{23}\right]  ds+\int\nolimits_{0}^{t}\widetilde{E}\left[
\widetilde{Z}_{s}\mathcal{A}_{s}^{2}\varphi\left(  \theta_{s}\right)
\vert
\mathcal{F}_{\infty}^{23}\right]  dX_{s}\\
&  +\int\nolimits_{0}^{t}\widetilde{E}\left[  \widetilde{Z}_{s}\mathcal{A}%
_{s}^{3}\varphi\left(  \theta_{s}\right)
\vert
\mathcal{F}_{\infty}^{23}\right]  dY_{s}\\
&  =\varsigma_{0}\left(  \varphi\right)  +\int\nolimits_{0}^{t}\varsigma
_{s}\left(  \mathcal{A}_{s}^{1}\varphi\right)  ds+\int\nolimits_{0}%
^{t}\varsigma_{s}\left(  \mathcal{A}_{s}^{2}\varphi\right)  dX_{s}%
+\int\nolimits_{0}^{t}\varsigma_{s}\left(  \mathcal{A}_{s}^{3}\varphi\right)
dY_{s}.%
\end{align*}

\end{proof}

The \textit{normalized} law $\pi$ of $\theta$ is given by the following

\begin{theorem}
\label{TheoremKushnerStratonovichEqSDE}(The Kushner-Stratonovich equation)

If $O\subseteq%
\mathbb{R}
$ is an open set containing $\left[  0,1\right]  $ and $\varphi\in
C^{2}\left(  O\right)  $ then $\forall t\geq0$,%
\begin{align}
\pi_{t}\left(  \varphi\right)   &  =\pi_{0}\left(  \varphi\right)
+\int\nolimits_{0}^{t}\pi_{s}\left(  \mathcal{A}_{s}^{1}\varphi\right)
ds+\int\nolimits_{0}^{t}\pi_{s}\left(  \varphi\right)  \left(  \pi_{s}%
^{2}\left(  \mathcal{A}_{s}^{2}\mathbf{1}\right)  +\pi_{s}^{2}\left(
\mathcal{A}_{s}^{3}\mathbf{1}\right)  \right)
ds\label{KushnerStratonovichEqSDE}\\
&  -\int\nolimits_{0}^{t}\left(  \pi_{s}\left(  \mathcal{A}_{s}^{2}%
\varphi\right)  \pi_{s}\left(  \mathcal{A}_{s}^{2}\mathbf{1}\right)  +\pi
_{s}\left(  \mathcal{A}_{s}^{3}\varphi\right)  \pi_{s}\left(  \mathcal{A}%
_{s}^{3}\mathbf{1}\right)  \right)  ds\nonumber\\
&  +\int\nolimits_{0}^{t}\left(  \pi_{s}\left(  \mathcal{A}_{s}^{2}%
\varphi\right)  -\pi_{s}\left(  \varphi\right)  \pi_{s}\left(  \mathcal{A}%
_{s}^{2}\mathbf{1}\right)  \right)  dX_{s}\nonumber\\
&  +\int\nolimits_{0}^{t}\left(  \pi_{s}\left(  \mathcal{A}_{s}^{3}%
\varphi\right)  -\pi_{s}\left(  \varphi\right)  \pi_{s}\left(  \mathcal{A}%
_{s}^{3}\mathbf{1}\right)  \right)  dY_{s}.\nonumber
\end{align}

\end{theorem}

\begin{proof}
Using Theorem \ref{TheoremZakaiEqSDE}
\begin{align*}
\zeta_{t}  &  =\varsigma_{t}\left(  \mathbf{1}\right) \\
&  =\varsigma_{0}\left(  \mathbf{1}\right)  +\int\nolimits_{0}^{t}%
\varsigma_{s}\left(  \mathcal{A}_{s}^{2}\mathbf{1}\right)  dX_{s}%
+\int\nolimits_{0}^{t}\varsigma_{s}\left(  \mathcal{A}_{s}^{3}\mathbf{1}%
\right)  dX_{s}\\
&  =1+\int\nolimits_{0}^{t}\zeta_{s}\pi_{s}\left(  \mathcal{A}_{s}%
^{2}\mathbf{1}\right)  dX_{s}+\int\nolimits_{0}^{t}\zeta_{s}\pi_{s}\left(
\mathcal{A}_{s}^{3}\mathbf{1}\right)  dY_{s}.%
\end{align*}
It follows that
\begin{align*}
\zeta_{t}  &  =\exp\left(  -\frac{1}{2}\int\nolimits_{0}^{t}\pi_{s}^{2}\left(
\mathcal{A}_{s}^{2}\mathbf{1}\right)  +\pi_{s}^{2}\left(  \mathcal{A}_{s}%
^{3}\mathbf{1}\right)  ds\right) \\
&  \times\exp\left(  \int\nolimits_{0}^{t}\pi_{s}\left(  \mathcal{A}_{s}%
^{2}\mathbf{1}\right)  dX_{s}+\int\nolimits_{0}^{t}\pi_{s}\left(
\mathcal{A}_{s}^{3}\mathbf{1}\right)  dY_{s}\right)
\end{align*}
and%
\begin{equation}
\zeta_{t}=\exp\left(  -\frac{1}{2}\int\nolimits_{0}^{t}\pi_{s}^{2}\left(
\mathcal{A}_{s}^{2}\mathbf{1}\right)  +\pi_{s}^{2}\left(  \mathcal{A}_{s}%
^{3}\mathbf{1}\right)  ds+\int\nolimits_{0}^{t}\pi_{s}\left(  \mathcal{A}%
_{s}^{2}\mathbf{1}\right)  dX_{s}+\int\nolimits_{0}^{t}\pi_{s}\left(
\mathcal{A}_{s}^{3}\mathbf{1}\right)  dY_{s}\right)  \text{.}
\label{ZakaiNormalizerEq}%
\end{equation}
We can also compute
\begin{align*}
d\left(  \frac{1}{\zeta_{t}}\right)   &  =\frac{1}{\zeta_{t}^{3}}\left(
\varsigma_{t}^{2}\left(  \mathcal{A}_{t}^{2}\mathbf{1}\right)  +\varsigma
_{t}^{2}\left(  \mathcal{A}_{t}^{3}\mathbf{1}\right)  \right)  dt-\frac
{1}{\zeta_{t}^{2}}\left(  \varsigma_{t}\left(  \mathcal{A}_{t}^{2}%
\mathbf{1}\right)  dX_{t}+\varsigma_{t}\left(  \mathcal{A}_{s}^{3}%
\mathbf{1}\right)  dY_{t}\right) \\
&  =\frac{1}{\zeta_{t}}\left(  \pi_{t}^{2}\left(  \mathcal{A}_{t}%
^{2}\mathbf{1}\right)  dt+\pi_{t}^{2}\left(  \mathcal{A}_{t}^{3}%
\mathbf{1}\right)  dt\right)  -\frac{1}{\zeta_{t}}\left(  \pi_{t}\left(
\mathcal{A}_{t}^{2}\mathbf{1}\right)  dX_{t}+\pi_{t}\left(  \mathcal{A}%
_{t}^{3}\mathbf{1}\right)  dY_{t}\right)
\end{align*}
and therefore,%
\begin{align*}
d\pi_{t}\left(  \varphi\right)   &  =d\left(  \frac{\varsigma_{t}\left(
\varphi\right)  }{\zeta_{t}}\right) \\
&  =\pi_{t}\left(  \mathcal{A}_{t}^{1}\varphi\right)  dt+\pi_{t}\left(
\mathcal{A}_{t}^{2}\varphi\right)  dX_{t}+\pi_{t}\left(  \mathcal{A}_{t}%
^{3}\varphi\right)  dY_{t}+\pi_{t}\left(  \varphi\right)  \left(  \pi_{t}%
^{2}\left(  \mathcal{A}_{t}^{2}\mathbf{1}\right)  dt+\pi_{t}^{2}\left(
\mathcal{A}_{t}^{3}\mathbf{1}\right)  dt\right) \\
&  -\pi_{t}\left(  \varphi\right)  \left(  \pi_{t}\left(  \mathcal{A}_{t}%
^{2}\mathbf{1}\right)  dX_{t}+\pi_{t}\left(  \mathcal{A}_{t}^{3}%
\mathbf{1}\right)  dY_{t}\right)  -\pi_{t}\left(  \mathcal{A}_{t}^{2}%
\varphi\right)  \pi_{t}\left(  \mathcal{A}_{t}^{2}\mathbf{1}\right)
dt-\pi_{t}\left(  \mathcal{A}_{t}^{3}\varphi\right)  \pi_{t}\left(
\mathcal{A}_{t}^{3}\mathbf{1}\right)  dt\\
&  =\left(  \pi_{t}\left(  \mathcal{A}_{t}^{1}\varphi\right)  +\pi_{t}\left(
\varphi\right)  \pi_{t}^{2}\left(  \mathcal{A}_{t}^{2}\mathbf{1}\right)
+\pi_{t}\left(  \varphi\right)  \pi_{t}^{2}\left(  \mathcal{A}_{t}%
^{3}\mathbf{1}\right)  \right)  dt-\pi_{t}\left(  \mathcal{A}_{t}^{2}%
\varphi\right)  \pi_{t}\left(  \mathcal{A}_{t}^{2}\mathbf{1}\right)  dt\\
&  +\pi_{t}\left(  \mathcal{A}_{t}^{3}\varphi\right)  \pi_{t}\left(
\mathcal{A}_{t}^{3}\mathbf{1}\right)  dt+\pi_{t}\left(  \mathcal{A}_{t}%
^{2}\varphi\right)  dX_{t}-\pi_{t}\left(  \varphi\right)  \pi_{t}\left(
\mathcal{A}_{t}^{2}\mathbf{1}\right)  dX_{t}+\pi_{t}\left(  \mathcal{A}%
_{t}^{3}\varphi\right)  dY_{t}\\
&  -\pi_{t}\left(  \varphi\right)  \pi_{t}\left(  \mathcal{A}_{t}%
^{3}\mathbf{1}\right)  dY_{t}.%
\end{align*}

\end{proof}

We end this subsection with the following useful

\begin{theorem}
\label{UniciteZakaiSDE}Assume that $P$-almost surely $\left(  \theta_{0}%
,v_{0},\rho_{0}\right)  \in\left]  0,1\right[  ^{3}$. Let $O$ in the Theorem
\ref{TheoremZakaiEqSDE} be bounded and $\varphi\in L^{2}\left(  O\right)  $.
If there is $T>0$ such that $f_{1}\in L^{\infty}\left(  \left]  0,T\right[
\times\Omega;W^{1,\infty}\left(  O;%
\mathbb{R}
\right)  \right)  $, $g_{1}\in L^{\infty}\left(  \left]  0,T\right[
\times\Omega;W^{2,\infty}\left(  O;%
\mathbb{R}
\right)  \right)  $, $f\in L^{\infty}\left(  \left]  0,T\right[  \times
\Omega;L_{loc}^{\infty}\left(
\mathbb{R}
\times O;%
\mathbb{R}
\right)  \right)  $ and $g\in L^{\infty}\left(  \left]  0,T\right[
\times\Omega;L_{loc}^{\infty}\left(
\mathbb{R}
^{2}\times O;%
\mathbb{R}
\right)  \right)  $ then the solution of the equation $\left(
\ref{ZakaiEqSDE}\right)  $ is unique and $\varsigma$ can be identified to an
element of $L^{\infty}\left(  \left]  0,T\right[  \times\Omega;H_{0}%
^{1}\left(  O\right)  \right)  $.
\end{theorem}

Before giving the a proof for the Theorem \ref{UniciteZakaiSDE} we first state
a result based on the Theorem 3.2.4 and Remark 3.2.6 in
\cite{burkholderpardoux} (pages 105 and 106).

\begin{theorem}
\label{TheoremGeneralUniciteEDPS}Let $\mathcal{X}$ and $\mathcal{Y}$ denote
two Hilbert spaces such that $\mathcal{X}$ is continuously embedded and dense
in $\mathcal{Y}$ which is identified with its dual space. Let $\mathcal{X}%
^{\prime}$ denote dual space of $\mathcal{X}$, $F\in L^{\infty}\left(  \left]
0,T\right[  \times\Omega;\mathcal{L}\left(  \mathcal{X},\mathcal{X}^{\prime
}\right)  \right)  $, $f\in L^{2}\left(  \left]  0,T\right[  \times
\Omega;\mathcal{X}^{\prime}\right)  $, $g\in L^{2}\left(  \left]  0,T\right[
\times\Omega;\mathcal{B}\left(  \mathcal{Y}^{N},\mathcal{Y}\right)  \right)
$, $G\in L^{\infty}\left(  \left]  0,T\right[  \times\Omega;\mathcal{L}\left(
\mathcal{X};\mathcal{L}\left(  \mathcal{Y}^{N},\mathcal{Y}\right)  \right)
\right)  $. We identify $\mathcal{L}\left(  \mathcal{Y}^{N},\mathcal{Y}%
\right)  $ with $\mathcal{Y}^{N}$ and assume that there are constants $c,C>0$
such that $\forall\psi\in\mathcal{X}$, $\forall t\in\left[  0,T\right]  $,
\begin{equation}
2\left\langle F\left(  t\right)  \psi,\psi\right\rangle +C\left\Vert
\psi\right\Vert _{\mathcal{Y}}^{2}\geq c\left\Vert \psi\right\Vert
_{\mathcal{X}}^{2}+\sum\nolimits_{i=1}^{N}\left\Vert G^{i}\left(
t,\psi\right)  \right\Vert _{\mathcal{Y}}^{2}\text{.}
\label{ConditionCoercivite}%
\end{equation}
Let $B$ denote the standard $\mathcal{Y}^{N}$-valued Brownian motion. Then
there is a unique $\Phi\in L^{2}\left(  \left]  0,T\right[  \times
\Omega;\mathcal{X}\right)  $ satisfying $\forall t\in\left[  0,T\right]  $,
\begin{equation}
\Phi\left(  t\right)  =\phi-\int\nolimits_{0}^{t}\left(  F\left(  s\right)
\Phi\left(  s\right)  +f\left(  s\right)  \right)  ds+\int\nolimits_{0}%
^{t}\left(  G\left(  s,\Phi\left(  s\right)  \right)  +g\left(  s\right)
\right)  dB_{s}%
\end{equation}

\end{theorem}

Looking at the proof of the Theorem \ref{TheoremGeneralUniciteEDPS} given in
\cite{burkholderpardoux} it is clear that it still remains true if $f$ and $G$
have a Lipschitz continuous dependance on $\Phi$.

\begin{proof}
(of the Theorem \ref{UniciteZakaiSDE})

We use the Theorem \ref{TheoremGeneralUniciteEDPS} setting $\mathcal{X}=H_{0}%
^{1}\left(  O\right)  $ and $\mathcal{Y}=L^{2}\left(  O\right)  $. Note that
the inclusions $\mathcal{X}\subseteq\mathcal{Y}\subseteq\mathcal{X}^{\prime}$
hold with continuous dense injections. Let $T>0$ be an abritrary fixed time
and $\mathcal{A}_{s}^{1\ast}$,$\mathcal{A}_{s}^{2\ast}$,$\mathcal{A}%
_{s}^{3\ast}$ be the respective adjoint operators of $\mathcal{A}_{s}^{1}%
$,$\mathcal{A}_{s}^{2}$,$\mathcal{A}_{s}^{3}:\mathcal{X}\rightarrow
\mathcal{Y}$. Omitting $\varphi$ in $\left(  \ref{ZakaiEqSDE}\right)  $ we
have the following SPDE
\begin{equation}
\varsigma_{t}=\varsigma_{0}+\int\nolimits_{0}^{t}\mathcal{A}_{s}^{1\ast
}\varsigma_{s}ds+\int\nolimits_{0}^{t}\mathcal{A}_{s}^{2\ast}\varsigma
_{s}dX_{s}+\int\nolimits_{0}^{t}\mathcal{A}_{s}^{3\ast}\varsigma_{s}%
dY_{s}\text{.} \label{ZakaiEqSDEMeasure}%
\end{equation}
Let $\psi\in\mathcal{X}$ and $t\in\left]  0,T\right]  $.
\begin{align*}
-2\left\langle \mathcal{A}_{t}^{1\ast}\psi,\psi\right\rangle  &
=-2\left\langle \psi,\mathcal{A}_{t}^{1}\psi\right\rangle \\
&  =-2\int\nolimits_{U}\left(  f_{1}\left(  t,x\right)  \psi^{\prime}\left(
x\right)  +g_{1}^{2}\left(  t,x\right)  \psi^{\prime\prime}\left(  x\right)
\right)  \psi\left(  x\right)  dx\\
&  =\int\nolimits_{U}f_{1}^{\prime}\left(  s,x\right)  \psi^{2}\left(
x\right)  dx+\int\nolimits_{U}\left(  g_{1}^{2}\left(  t,x\right)  \psi\left(
x\right)  \right)  ^{\prime}\psi^{\prime}\left(  x\right)  dx\\
&  =\int\nolimits_{U}f_{1}^{\prime}\left(  s,x\right)  \psi^{2}\left(
x\right)  dx+\int\nolimits_{U}g_{1}^{2}\left(  t,x\right)  \left(
\psi^{\prime}\left(  x\right)  \right)  ^{2}dx+\int\nolimits_{U}\left(
g_{1}^{2}\left(  t,x\right)  \right)  ^{\prime}\psi\left(  x\right)
\psi^{\prime}\left(  x\right)  dx\\
&  =\int\nolimits_{U}\left(  f_{1}^{\prime}\left(  s,x\right)  -\frac{1}%
{2}\left(  g_{1}^{2}\left(  t,x\right)  \right)  ^{\prime\prime}\right)
\psi^{2}\left(  x\right)  dx+\int\nolimits_{U}g_{1}^{2}\left(  t,x\right)
\left(  \psi^{\prime}\left(  x\right)  \right)  ^{2}dx\\
&  \geq-C_{1}^{T,O}\left\Vert \psi\right\Vert _{\mathcal{Y}}^{2}+C_{2}%
^{T,O}\left\Vert \psi^{\prime}\right\Vert _{\mathcal{Y}}^{2}%
\end{align*}
with
\[
C_{1}^{T,O}=\underset{x\in U,s\in\left]  0,T\right]  }{ess\sup}\left\vert
f_{1}^{\prime}\left(  s,x\right)  -\frac{1}{2}\left(  g_{1}^{2}\left(
t,x\right)  \right)  ^{\prime\prime}\right\vert \geq0
\]
and
\[
C_{2}^{T,O}=\underset{t\in\left]  0,T\right]  }{\inf\left\{  \delta_{1}\left(
t\right)  \right\}  }C_{3}^{T,O}>0
\]
The existence of $C_{3}^{T,O}$ is guaranteed because $\psi\in\mathcal{X\mapsto
}\left\Vert \kappa_{1}\left(  x\right)  \psi^{\prime}\right\Vert
_{\mathcal{Y}}$ is a norm equivalent to the usual norm\footnote{See
Proposition 8.13 in \cite{brezis} (page 218) on Poincar\'{e}'s inequality and
the open mapping Theorem 2.6 (page 35 ).} since the Assumption
\ref{HypothesisFilter1_8} is satisfied and we necessarily have $\left\Vert
\kappa_{1}\left(  x\right)  \psi^{\prime}\right\Vert _{\mathcal{Y}}^{2}>0$
when $\left\Vert \psi^{\prime}\right\Vert _{\mathcal{Y}}^{2}>0$.
$\mathcal{A}_{t}^{2}$ and $\mathcal{A}_{t}^{3}$ are self-adjoint and
\[
\left\Vert \mathcal{A}_{t}^{2}\psi\left(  x\right)  \right\Vert _{\mathcal{Y}%
}^{2}+\left\Vert \mathcal{A}_{t}^{3}\psi\left(  x\right)  \right\Vert
_{\mathcal{Y}}^{2}\leq C_{4}^{T,O}\int\nolimits_{U}\psi^{2}\left(  x\right)
dx
\]
with
\[
C_{4}^{T,O}=\underset{x\in U,t\in\left]  0,T\right]  }{ess\sup}\left(
\frac{f^{2}\left(  t,\overline{X}_{t},\theta_{t}\right)  }{\delta_{2}%
^{4}\left(  t\right)  }+\frac{g^{2}\left(  t,\overline{X}_{t},\overline{Y}%
_{t},\theta_{t}\right)  }{\delta_{3}^{4}\left(  t\right)  }\right)  \geq0
\]
The existence of $C_{4}^{T,O}$ is due to the boundedness of $f_{i}$, $\forall
i\in\left\{  2,3\right\}  $ and Lemma \ref{RemainInteriorSDE}. Hence, the
condition $\left(  \ref{ConditionCoercivite}\right)  $ is satisfied with
$C>C_{1}^{T,O}+C_{4}^{T,O}$ and $c\leq\min\left\{  C_{1}^{T,O}+C_{4}%
^{T,O},C_{2}^{T,O}\right\}  $. Thefore $\varsigma_{t}\in\mathcal{X}$ is the
unique solution of $\left(  \ref{ZakaiEqSDEMeasure}\right)  $.
\end{proof}

The Theorem \ref{UniciteZakaiSDE} gives conditions under which $\left(
\varsigma_{t}\right)  _{t\geq0}$ and therefore $\left(  \pi_{t}\right)
_{t\geq0}$ are uniquely defined by their respective equations. Note that the
domain of $\left(  \varsigma_{t}\right)  _{t\geq0}$ and $\left(  \pi
_{t}\right)  _{t\geq0}$ has been extended from $C^{2}\left(  O\right)  $ to
$L^{2}\left(  O\right)  $.

\section{State estimation with discrete time observations}\label{OtherIssues}

In this section we, consider a realistic case where observations are discretely
made with respect to an increasing sequence of nonnegative stopping times
$\left(  \tau_{n}\right)  _{n\in%
\mathbb{N}
}$ such that $\underset{n\rightarrow\infty}{\lim}\tau_{n}=\infty$. That
situation occurs frequently when following a phenomenon since it is difficult
to collect data continuously. The problem here is to find $E\left[
\varphi\left(  \theta_{t}\right)
\vert
\mathcal{F}_{n}^{23}\right]  $, $\forall t\in\left[  \tau_{n},\tau
_{n+1}\right[  $, $\forall\varphi\in L_{loc}^{\infty}\left(
\mathbb{R}
\right)  $. We may mention here that $\left(  \mathcal{F}_{n}^{23}\right)  $
is the discrete filtration generated by the process $\left(  \left(
v_{\tau_{n}},\rho_{\tau_{n}}\right)  \right)  _{n\in%
\mathbb{N}
}$, parameters and all $P$-null sets. We can distinguish two cases. Indeed, if
$t\in\left]  \tau_{n},\tau_{n+1}\right[  $, assuming that the law of
$\theta_{\tau_{n}}$ is known then we have to solve a prediction problem. The
case $t=\tau_{n}$ corresponds to a discrete filtering problem. We study those
two situations in the following.

\subsection{Prediction problem}\label{SubsectionPredictionSDE}

\qquad In this subsection, we assume that at each time $t\geq0$ only the
observations $\left(  \left(  v_{s\wedge\tau},\rho_{s\wedge\tau}\right)
\right)  _{0\leq s\leq t}$ are available with $\tau$ a stopping time. To deal
with the prediction problem we can just assume as in \cite{burkholderpardoux}
that after $\tau$ the observations are reduced to a new independent Brownian
motion. That is $\forall t\geq0$,
\begin{equation}
\widehat{X}_{t}=\overline{X}_{t}+\int\nolimits_{0}^{t}\mathbf{1}_{\left\{
s\leq\tau\right\}  }\delta_{2}\left(  s\right)  dB_{s}^{2}+W_{t}%
^{2}-W_{t\wedge\tau}^{2}%
\end{equation}
and%
\begin{equation}
\widehat{Y}_{t}=\overline{Y}_{t}+\int\nolimits_{0}^{t}\mathbf{1}_{\left\{
s\leq\tau\right\}  }\delta_{3}\left(  s\right)  dB_{s}^{3}+W_{t}%
^{3}-W_{t\wedge\tau}^{3}\text{.}%
\end{equation}
where $\left(  W^{2},W^{3}\right)  =W$ is an independent two-dimensional
Brownian motion. We can easily check that $\left(  \widehat{X}_{t\wedge\tau
},\widehat{Y}_{t\wedge\tau}\right)  =\left(  X_{t\wedge\tau},Y_{t\wedge\tau
}\right)  $ and $\left(  \widehat{X}_{t\vee\tau},\widehat{Y}_{t\vee\tau
}\right)  =\left(  X_{\tau}+W_{t}^{2}-W_{t\wedge\tau}^{2},Y_{\tau}+W_{t}%
^{3}-W_{t\wedge\tau}^{3}\right)  $. That permits us to use a similar approach
with the Subsection \ref{SubsectionContinuousObservationSDE}. If $\left(
\widehat{\mathcal{F}}_{t}^{23}\right)  $ is the filtration generated by
$\left(  \widehat{X},\widehat{Y}\right)  $, parameters and all $P$-null sets
then by the independence of $W$ we have $E\left[  \theta_{t}%
\vert
\mathcal{F}_{t\wedge\tau}^{23}\right]  =E\left[  \theta_{t}%
\vert
\widehat{\mathcal{F}}_{t}^{23}\right]  $. Let $\forall\varphi\in L^{\infty
}\left(
\mathbb{R}
;%
\mathbb{R}
\right)  ,$ $\forall t\geq0,$
\begin{equation}
\widehat{Z}_{t}=\widetilde{Z}_{t\wedge\tau}%
\end{equation}%
\begin{equation}
\widehat{\pi}_{t}\left(  \varphi\right)  =E\left[  \varphi\left(  \theta
_{t}\right)
\vert
\widehat{\mathcal{F}}_{t}^{23}\right]
\end{equation}%
\begin{equation}
\widehat{\zeta}_{t}=\widetilde{E}\left[  \widehat{Z}_{t}%
\vert
\widehat{\mathcal{F}}_{t}^{23}\right]  =\zeta_{t\wedge\tau}%
\end{equation}%
\begin{equation}
\widehat{\varsigma}_{t}\left(  \varphi\right)  =\widetilde{E}\left[
\widehat{Z}_{t}\varphi\left(  \theta_{t}\right)
\vert
\widehat{\mathcal{F}}_{t}^{23}\right]  =\widehat{\zeta}_{t}\widehat{\pi}%
_{t}\left(  \varphi\right)
\end{equation}
Note that if $t\leq\tau$ then $\widehat{\varsigma}_{t}=\varsigma_{t}$ and
$\zeta_{t}=\widehat{\zeta}_{t}$.

The dynamics of the unnormalized law $\widehat{\varsigma}$ is given by the

\begin{theorem}
\label{TheoremZakaiPredEqSDE}(The Zakai prediction equation)

If $O\subseteq%
\mathbb{R}
$ is an open set containing $\left[  0,1\right]  $ and $\varphi\in
C^{2}\left(  O\right)  $ then $\forall t>\tau$,
\begin{align}
\widehat{\varsigma}_{t}\left(  \varphi\right)   &  =\varsigma_{0}\left(
\varphi\right)  \varsigma_{\tau}\left(  \mathbf{1}\right)  +\int
\nolimits_{0}^{t}\widehat{\varsigma}_{s}\left(  \mathcal{A}_{s}^{1}%
\varphi\right)  ds+\int\nolimits_{0}^{\tau}\varsigma_{s}\left(  \mathcal{A}%
_{s}^{2}\varphi-\varsigma_{0}\left(  \varphi\right)  \mathcal{A}_{s}%
^{2}\mathbf{1}\right)  dX_{s}\label{ZakaiPredEqSDE}\\
&  +\int\nolimits_{0}^{\tau}\varsigma_{s}\left(  \mathcal{A}_{s}^{3}%
\varphi-\varsigma_{0}\left(  \varphi\right)  \mathcal{A}_{s}^{3}%
\mathbf{1}\right)  dY_{s}\nonumber
\end{align}
where $\forall i\in\left\{  1,2,3\right\}  $, $\mathcal{A}_{t}^{1}$ and
$\varsigma$ are given in the Theorem \ref{TheoremZakaiEqSDE}.
\end{theorem}

\begin{proof}
Using the Lemma \ref{LemmaPermutIntegralExpect} and integration by part
formula we have
\begin{align*}
\widehat{\varsigma}_{t}\left(  \varphi\right)   &  =\widetilde{E}\left[
\widehat{Z}_{t}\varphi\left(  \theta_{t}\right)
\vert
\widehat{\mathcal{F}}_{t}^{23}\right] \\
&  =\widetilde{E}\left[  \widehat{Z}_{t}\varphi\left(  \theta_{0}\right)
\vert
\widehat{\mathcal{F}}_{\infty}^{23}\right]  +\widetilde{E}\left[  \widehat
{Z}_{t}\int\nolimits_{0}^{t}\mathcal{A}_{s}^{1}\varphi\left(  \theta
_{s}\right)  ds%
\vert
\widehat{\mathcal{F}}_{\infty}^{23}\right]  +\widetilde{E}\left[  \widehat
{Z}_{t}\int\nolimits_{0}^{t}g^{1}\left(  \theta_{s}\right)  D\varphi\left(
\theta_{s}\right)  dB_{s}^{1}%
\vert
\widehat{\mathcal{F}}_{\infty}^{23}\right] \\
&  =\widehat{\varsigma}_{0}\left(  \varphi\right)  \widehat{\varsigma}%
_{t}\left(  \mathbf{1}\right)  +\widetilde{E}\left[  \widetilde{Z}_{\tau}%
\int\nolimits_{0}^{\tau}\mathcal{A}_{s}^{1}\varphi\left(  \theta_{s}\right)
ds%
\vert
\widehat{\mathcal{F}}_{\infty}^{23}\right]  +\int\nolimits_{\tau}%
^{t}\widetilde{E}\left[  \widetilde{Z}_{s}\mathcal{A}_{s}^{1}\varphi\left(
\theta_{s}\right)
\vert
\widehat{\mathcal{F}}_{\infty}^{23}\right]  ds\\
&  =\varsigma_{0}\left(  \varphi\right)  \varsigma_{\tau}\left(
\mathbf{1}\right)  +\widetilde{E}\left[  \int\nolimits_{0}^{\tau}\widehat
{Z}_{s}\mathcal{A}_{s}^{1}\varphi\left(  \theta_{s}\right)  ds%
\vert
\widehat{\mathcal{F}}_{\infty}^{23}\right]  +\widetilde{E}\left[
\int\nolimits_{0}^{\tau}\mathcal{A}_{s}^{2}\widetilde{Z}_{s}\int
\nolimits_{0}^{s}\mathcal{A}_{r}^{1}\varphi\left(  \theta_{r}\right)  drdX_{s}%
\vert
\widehat{\mathcal{F}}_{\infty}^{23}\right] \\
&  +\widetilde{E}\left[  \int\nolimits_{0}^{\tau}\mathcal{A}_{s}^{3}%
\widetilde{Z}_{s}\int\nolimits_{0}^{s}\mathcal{A}_{r}^{1}\varphi\left(
\theta_{r}\right)  drdY_{s}%
\vert
\widehat{\mathcal{F}}_{\infty}^{23}\right]  +\int\nolimits_{\tau}%
^{t}\widetilde{E}\left[  \widehat{Z}_{s}\mathcal{A}_{s}^{1}\varphi\left(
\theta_{s}\right)
\vert
\widehat{\mathcal{F}}_{\infty}^{23}\right]  ds
\end{align*}%
\begin{align*}
&  =\varsigma_{0}\left(  \varphi\right)  \varsigma_{\tau}\left(
\mathbf{1}\right)  +\int\nolimits_{0}^{t}\widetilde{E}\left[  \widehat{Z}%
_{s}\mathcal{A}_{s}^{1}\varphi\left(  \theta_{s}\right)
\vert
\widehat{\mathcal{F}}_{\infty}^{23}\right]  ds+\int\nolimits_{0}^{\tau
}\widetilde{E}\left[  \left(  \varphi\left(  \theta_{s}\right)  -\varphi
\left(  \theta_{0}\right)  \right)  \mathcal{A}_{s}^{2}\widetilde{Z}_{s}%
\vert
\widehat{\mathcal{F}}_{\infty}^{23}\right]  dX_{s}\\
&  +\int\nolimits_{0}^{\tau}\widetilde{E}\left[  \left(  \varphi\left(
\theta_{s}\right)  -\varphi\left(  \theta_{0}\right)  \right)  \mathcal{A}%
_{s}^{3}\widetilde{Z}_{s}%
\vert
\widehat{\mathcal{F}}_{\infty}^{23}\right]  dY_{s}\\
&  =\varsigma_{0}\left(  \varphi\right)  \varsigma_{\tau}\left(
\mathbf{1}\right)  +\int\nolimits_{0}^{t}\widehat{\varsigma}_{s}\left(
\mathcal{A}_{s}^{1}\varphi\right)  ds+\int\nolimits_{0}^{\tau}\varsigma
_{s}\left(  \mathcal{A}_{s}^{2}\varphi-\varsigma_{0}\left(  \varphi\right)
\mathcal{A}_{s}^{2}\mathbf{1}\right)  dX_{s}\\
&  +\int\nolimits_{0}^{\tau}\varsigma_{s}\left(  \mathcal{A}_{s}^{3}%
\varphi-\varsigma_{0}\left(  \varphi\right)  \mathcal{A}_{s}^{3}%
\mathbf{1}\right)  dY_{s}%
\end{align*}

\end{proof}

The dynamics of the normalized law $\widehat{\pi}$ is given by the

\begin{theorem}
\label{TheoremKushnerStratonovichPredEqSDE}(The Kushner-Stratonovich
prediction equation)

If $O\subseteq%
\mathbb{R}
$ is an open set containing $\left[  0,1\right]  $ and $\varphi\in
C^{2}\left(  O\right)  $ then $\forall t>\tau$,%
\begin{align}
\widehat{\pi}_{t}\left(  \varphi\right)   &  =\pi_{0}\left(  \varphi\right)
+\int\nolimits_{0}^{t}\widehat{\pi}_{s}\left(  \mathcal{A}_{s}^{1}%
\varphi\right)  ds-\int\nolimits_{0}^{\tau}\pi_{s}\left(  \mathcal{A}_{s}%
^{2}\varphi-\pi_{0}\left(  \varphi\right)  \mathcal{A}_{s}^{2}\mathbf{1}%
\right)  \pi_{s}\left(  \mathcal{A}_{s}^{2}\mathbf{1}\right)
ds\label{KushnerStratonovichPredEqSDE}\\
&  -\int\nolimits_{0}^{\tau}\pi_{s}\left(  \mathcal{A}_{s}^{3}\varphi-\pi
_{0}\left(  \varphi\right)  \mathcal{A}_{s}^{3}\mathbf{1}\right)  \pi
_{s}\left(  \mathcal{A}_{s}^{3}\mathbf{1}\right)  ds\nonumber
\end{align}

\end{theorem}

\begin{proof}
Let consider the probability $\widetilde{P}$ and $t>\tau$. Using integration
by part formula we have
\begin{align*}
\frac{\widehat{\varsigma}_{t}}{\widehat{\zeta}_{t}}\left(  \varphi\right)   &
=\varsigma_{0}\left(  \varphi\right)  \frac{\varsigma_{\tau}\left(
\mathbf{1}\right)  }{\zeta_{\tau}}+\frac{1}{\zeta_{\tau}}\int\nolimits_{\tau
}^{t}\widehat{\varsigma}_{s}\left(  \mathcal{A}_{s}^{1}\varphi\right)
ds+\frac{1}{\zeta_{\tau}}\int\nolimits_{0}^{\tau}\widehat{\varsigma}%
_{s}\left(  \mathcal{A}_{s}^{1}\varphi\right)  ds\\
&  +\frac{1}{\zeta_{\tau}}\int\nolimits_{0}^{\tau}\varsigma_{s}\left(
\mathcal{A}_{s}^{2}\varphi-\varsigma_{0}\left(  \varphi\right)  \mathcal{A}%
_{s}^{2}\mathbf{1}\right)  dX_{s}+\frac{1}{\zeta_{\tau}}\int\nolimits_{0}%
^{\tau}\varsigma_{s}\left(  \mathcal{A}_{s}^{3}\varphi-\varsigma_{0}\left(
\varphi\right)  \mathcal{A}_{s}^{3}\mathbf{1}\right)  dY_{s}\\
&  =\varsigma_{0}\left(  \varphi\right)  -\int\nolimits_{0}^{\tau}\pi
_{s}\left(  \mathcal{A}_{s}^{2}\varphi-\varsigma_{0}\left(  \varphi\right)
\mathcal{A}_{s}^{2}\mathbf{1}\right)  \pi_{s}\left(  \mathcal{A}_{s}%
^{2}\mathbf{1}\right)  ds\\
&  -\int\nolimits_{0}^{\tau}\pi_{s}\left(  \mathcal{A}_{s}^{3}\varphi
-\varsigma_{0}\left(  \varphi\right)  \mathcal{A}_{s}^{3}\mathbf{1}\right)
\pi_{s}\left(  \mathcal{A}_{s}^{3}\mathbf{1}\right)  ds\\
&  +\int\nolimits_{0}^{\tau}\left(  \int\nolimits_{0}^{s}\varsigma_{r}\left(
\mathcal{A}_{r}^{2}\varphi-\varsigma_{0}\left(  \varphi\right)  \mathcal{A}%
_{r}^{2}\mathbf{1}\right)  dr\right)  d\left(  \frac{1}{\zeta_{s}}\right) \\
&  +\int\nolimits_{0}^{\tau}\left(  \int\nolimits_{0}^{s}\varsigma_{r}\left(
\mathcal{A}_{r}^{3}\varphi-\varsigma_{0}\left(  \varphi\right)  \mathcal{A}%
_{r}^{3}\mathbf{1}\right)  dr\right)  d\left(  \frac{1}{\zeta_{s}}\right) \\
&  +\int\nolimits_{0}^{t}\widehat{\pi}_{s}\left(  \mathcal{A}_{s}^{1}%
\varphi\right)  ds+\int\nolimits_{0}^{\tau}\left(  \int\nolimits_{0}%
^{s}\varsigma_{r}\left(  \mathcal{A}_{r}^{1}\varphi\right)  dr\right)
d\left(  \frac{1}{\zeta_{s}}\right)
\end{align*}%
\begin{align*}
&  +\int\nolimits_{0}^{\tau}\pi_{s}\left(  \mathcal{A}_{s}^{2}\varphi
-\varsigma_{0}\left(  \varphi\right)  \mathcal{A}_{s}^{2}\mathbf{1}\right)
dX_{s}+\int\nolimits_{0}^{\tau}\pi_{s}\left(  \mathcal{A}_{s}^{3}%
\varphi-\varsigma_{0}\left(  \varphi\right)  \mathcal{A}_{s}^{3}%
\mathbf{1}\right)  dY_{s}\\
&  =\varsigma_{0}\left(  \varphi\right)  +\int\nolimits_{0}^{t}\widehat{\pi
}_{s}\left(  \mathcal{A}_{s}^{1}\varphi\right)  ds-\int\nolimits_{0}^{\tau}%
\pi_{s}\left(  \mathcal{A}_{s}^{2}\varphi-\varsigma_{0}\left(  \varphi\right)
\mathcal{A}_{s}^{2}\mathbf{1}\right)  \pi_{s}\left(  \mathcal{A}_{s}%
^{2}\mathbf{1}\right)  ds\\
&  -\int\nolimits_{0}^{\tau}\pi_{s}\left(  \mathcal{A}_{s}^{3}\varphi
-\varsigma_{0}\left(  \varphi\right)  \mathcal{A}_{s}^{3}\mathbf{1}\right)
\pi_{s}\left(  \mathcal{A}_{s}^{3}\mathbf{1}\right)  ds+\int\nolimits_{0}%
^{\tau}\left(  \varsigma_{s}\left(  \varphi\right)  -\varsigma_{0}\left(
\varphi\right)  \right)  d\left(  \frac{1}{\zeta_{s}}\right) \\
&  -\int\nolimits_{0}^{\tau}\left(  \int\nolimits_{0}^{s}\varsigma_{0}\left(
\varphi\right)  \varsigma_{r}\left(  \mathcal{A}_{r}^{2}\mathbf{1}\right)
dr\right)  d\left(  \frac{1}{\zeta_{s}}\right)  -\int\nolimits_{0}^{\tau
}\left(  \int\nolimits_{0}^{s}\varsigma_{0}\left(  \varphi\right)
\varsigma_{r}\left(  \mathcal{A}_{r}^{3}\mathbf{1}\right)  dr\right)  d\left(
\frac{1}{\zeta_{s}}\right) \\
&  =\varsigma_{0}\left(  \varphi\right)  +\int\nolimits_{0}^{t}\widehat{\pi
}_{s}\left(  \mathcal{A}_{s}^{1}\varphi\right)  ds-\int\nolimits_{0}^{\tau}%
\pi_{s}\left(  \mathcal{A}_{s}^{2}\varphi-\varsigma_{0}\left(  \varphi\right)
\mathcal{A}_{s}^{2}\mathbf{1}\right)  \pi_{s}\left(  \mathcal{A}_{s}%
^{2}\mathbf{1}\right)  ds\\
&  -\int\nolimits_{0}^{\tau}\pi_{s}\left(  \mathcal{A}_{s}^{3}\varphi
-\varsigma_{0}\left(  \varphi\right)  \mathcal{A}_{s}^{3}\mathbf{1}\right)
\pi_{s}\left(  \mathcal{A}_{s}^{3}\mathbf{1}\right)  ds
\end{align*}
Since $\varsigma_{0}=\pi_{0}$ the results follows.
\end{proof}

The following existence and uniqueness result holds.

\begin{theorem}
\label{UniciteZakaiPredSDE}Assume that $P$-almost surely $\left(  \theta
_{0},v_{0},\rho_{0}\right)  \in\left]  0,1\right[  ^{3}$. Let $O$ in the
Theorem \ref{TheoremZakaiPredEqSDE} be bounded and $\varphi\in L^{2}\left(
O\right)  $. If there is $T>0$ such that $f_{1}\in L^{\infty}\left(  \left]
0,T\right[  \times\Omega;W^{1,\infty}\left(  O;%
\mathbb{R}
\right)  \right)  $, $g_{1}\in L^{\infty}\left(  \left]  0,T\right[
\times\Omega;\right.  $ $\left.  W^{2,\infty}\left(  O;%
\mathbb{R}
\right)  \right)  $, $f\in L^{\infty}\left(  \left]  0,T\right[  \times
\Omega;L_{loc}^{\infty}\left(
\mathbb{R}
\times\left[  0,v_{\max}\right]  \times O;%
\mathbb{R}
\right)  \right)  $ and $g\in L^{\infty}\left(  \left]  0,T\right[
\times\Omega;L_{loc}^{\infty}\left(
\mathbb{R}
^{2}\times\left[  0,1\right]  \right.  \right.  $ $\left.  \left.  \times O;%
\mathbb{R}
\right)  \right)  $ then the solution of the equation $\left(
\ref{ZakaiPredEqSDE}\right)  $ is unique and $\widehat{\varsigma}$ can be
identified to an element of $L^{\infty}\left(  \left]  0,T\right[
\times\Omega;H_{0}^{1}\left(  O\right)  \right)  $.
\end{theorem}

\begin{proof}
The proof is similar to the one of the Theorem \ref{TheoremZakaiEqSDE}. We
also use the Theorem \ref{TheoremGeneralUniciteEDPS} setting $\mathcal{X}%
=H_{0}^{1}\left(  O\right)  $ and $\mathcal{Y}=L^{2}\left(  O\right)  $.
Omitting $\varphi$ in $\left(  \ref{ZakaiPredEqSDE}\right)  $ we have the
Fokker-Planck type equation
\begin{align}
\widehat{\varsigma}_{t}  &  =\varsigma_{\tau}\left(  \mathbf{1}\right)
\varsigma_{0}+\int\nolimits_{0}^{t}\mathcal{A}_{s}^{1\ast}\widehat{\varsigma
}_{s}ds+\int\nolimits_{0}^{\tau}\left(  \mathcal{A}_{s}^{2\ast}\varsigma
_{s}-\varsigma_{s}\left(  \mathcal{A}_{s}^{2}\mathbf{1}\right)  \varsigma
_{0}\right)  dX_{s}\label{ZakaiPredEqSDEMeasure}\\
&  +\int\nolimits_{0}^{\tau}\left(  \mathcal{A}_{s}^{3\ast}\varsigma
_{s}-\varsigma_{s}\left(  \mathcal{A}_{s}^{3}\mathbf{1}\right)  \varsigma
_{0}\right)  dY_{s}\nonumber
\end{align}
We already know properties of the operators $\mathcal{A}_{s}^{i}$ given by
$\left(  \ref{Operator1Filter1}\right)  -\left(  \ref{Operator3}\right)  $ and
their adjoints. Conditions of the Theorem \ref{TheoremGeneralUniciteEDPS} are
satisfied and the result follows.
\end{proof}

\subsection{Discrete filtering problem}\label{SubsectionDiscreteFiltering}

\qquad In this subsection, we consider the discrete filtering problem mentioned
above. The process $\left(  \theta,v,\rho\right)  $ is markovian and therefore
the discrete process $\left(  \left(  \theta_{\tau_{n}},v_{\tau_{n}}%
,\rho_{\tau_{n}}\right)  \right)  _{n\in%
\mathbb{N}
}$ is a Markov chain. To achieve our objective we will make some
approximations in order to have a discrete filtering problem. To make simple
the notations, we introduce when there is not ambiguity the index $n$ to play
the role of $\tau_{n}$. If $\Delta_{n}$ denotes the difference $\tau
_{n+1}-\tau_{n}$ and it is sufficiently small then the following
approximations hold for a given $\vartheta\in\left[  0,1\right]  $:
\begin{align}
\Delta\theta_{n} &  \equiv\theta_{n+1}-\theta_{n}\\
&  \simeq\Delta\tau_{n}f_{1,n}\left(  \left(  1-\vartheta\right)  \theta
_{n}+\vartheta\theta_{n+1}\right)  +\sqrt{\Delta\tau_{n}}g_{1,n}\left(
\theta_{n}\right)  \xi_{1,n}\text{,}\nonumber
\end{align}%
\begin{align}
\Delta X_{n} &  \equiv X_{n+1}-X_{n}\\
&  \simeq\sqrt{\Delta\tau_{n}}\delta_{2,n}\xi_{2,n}\nonumber\\
&  +\Delta\tau_{n}f_{n}\left(  \left(  1-\vartheta\right)  \overline{X}%
_{n}+\vartheta\overline{X}_{n+1},\left(  1-\vartheta\right)  \theta
_{n}+\vartheta\theta_{n+1}\right)  \nonumber
\end{align}
and
\begin{align}
\Delta Y_{n} &  \equiv Y_{n+1}-Y_{n}\\
&  \simeq\sqrt{\Delta\tau_{n}}\delta_{3,n}\xi_{3,n}\nonumber\\
&  +\Delta\tau_{n}g_{n}\left(  \left(  1-\vartheta\right)  \overline{X}%
_{n}+\vartheta\overline{X}_{n+1},\left(  1-\vartheta\right)  \overline{Y}%
_{n}+\vartheta\overline{Y}_{n+1},\left(  1-\vartheta\right)  \theta
_{n}+\vartheta\theta_{n+1}\right)  \nonumber
\end{align}
with $\left(  \xi_{n}\right)  _{n\in%
\mathbb{N}
}=\left(  \left(  \xi_{1,n},\xi_{2,n},\xi_{3,n}\right)  \right)  _{n\in%
\mathbb{N}
}$ a sequence of independent indentically distributed centred and normalized
gaussian vectors. The use of the term $\vartheta$ corresponds to the
well-known \textit{theta method} in the large literature of numerical
analysis. It is justified by the fact that the mathematical expectation of
$\left(  \left(  \theta_{t},v_{t},\rho_{t}\right)  \right)  _{t\geq0}$ is
differentiable and we can use the finite increments formula. That cannot be
applied to the Brownian term if we want to keep safe the properties of the
It\^{o} integral. We refer to the works in
\cite{higham,hutzenthaler,kloeden,saito} and references therein to know
further about stochastic numerical schemes.

Let $\forall n\in%
\mathbb{N}
$ , $\forall x\in%
\mathbb{R}
$, $\overline{Z}_{0}=\overline{\Lambda}_{0}=1$,
\begin{align}
\overline{\Lambda}_{n}\left(  x,y,z\right)   &  =\exp\left(  -\Delta\tau
_{n}f_{n}^{2}\left(  \left(  1-\vartheta\right)  \overline{X}_{n}%
+\vartheta\overline{X}_{n+1}\right.  \right.  \left.  \left.  ,\left(
1-\vartheta\right)  \theta_{n}+\vartheta x\right)  /2\delta_{2,n}^{2}\right)
\\
&  \;\;\;\times\exp\left(  -\Delta\tau_{n}g_{n}^{2}\left(  \left(
1-\vartheta\right)  \overline{X}_{n}+\vartheta\overline{X}_{n+1},\left(
1-\vartheta\right)  \overline{Y}_{n}+\vartheta\overline{Y}_{n+1},\right.
\right.  \left.  \left.  \left(  1-\vartheta\right)  \theta_{n}+\vartheta
x\right)  /2\delta_{3,n}^{2}\right) \nonumber\\
&  \;\;\;\times\exp\left(  yf_{n}\left(  \left(  1-\vartheta\right)
\overline{X}_{n}+\vartheta\overline{X}_{n+1}\right.  \right.  \left.  \left.
,\left(  1-\vartheta\right)  \theta_{n}+\vartheta x\right)  /\delta_{2,n}%
^{2}\right) \nonumber\\
&  \;\;\;\times\exp\left(  zg_{n}\left(  \left(  1-\vartheta\right)
\overline{X}_{n}+\vartheta\overline{X}_{n+1},\left(  1-\vartheta\right)
\overline{Y}_{n}+\vartheta\overline{Y}_{n+1},\right.  \right.  \left.  \left.
\left(  1-\vartheta\right)  \theta_{n}+\vartheta x\right)  /\delta_{3,n}%
^{2}\right) \nonumber
\end{align}
and%
\begin{equation}
\overline{Z}_{n}=\prod\nolimits_{i=0}^{n}\overline{\Lambda}_{i}\left(
\theta_{i+1},\Delta X_{i},\Delta Y_{i}\right)
\end{equation}
By the Girsanov theorem, the discrete process $\left(  \overline{Z}%
_{n}\right)  _{n\in%
\mathbb{N}
}$ is an $\left(  \mathcal{F}_{n}^{23}\right)  $-martingale and there is a
probability $\overline{P}$ such that
\begin{equation}
\overline{Z}_{n}=\frac{dP%
\vert
_{\mathcal{F}_{n}^{23}}}{d\overline{P}%
\vert
_{\mathcal{F}_{n}^{23}}}\text{.}%
\end{equation}
Note that parameters of the model and $\theta$ keep the same law either under
$P$ or $\overline{P}$. Moreover, under $\overline{P}$ parameters and $\theta$
are independent with the process $\left(  \left(  \frac{\Delta X_{n}}%
{\sqrt{\Delta_{n}}\delta_{2,n}},\frac{\Delta Y_{n}}{\sqrt{\Delta_{n}}%
\delta_{3,n}}\right)  \right)  _{n\in%
\mathbb{N}
}$ which is a sequence of independent identically distributed centred and
normalized gaussian vectors.

Let also define $\forall n\in%
\mathbb{N}
$ , $\forall x\in%
\mathbb{R}
$, $\forall\varphi\in L^{\infty}\left(
\mathbb{R}
;%
\mathbb{R}
\right)  $%
\begin{equation}
\zeta_{n}=\overline{E}\left[  \overline{Z}_{n}%
\vert
\mathcal{F}_{n}^{23}\right],
\end{equation}%
\begin{equation}
\pi_{n}\left(  \varphi\right)  =E\left[  \varphi\left(  \theta_{n}\right)
\vert
\mathcal{F}_{n}^{23}\right],
\end{equation}%
\begin{equation}
\varsigma_{n,n+1}\left(  \varphi\right)  =\overline{E}\left[  \overline{Z}%
_{n}\varphi\left(  \theta_{n+1}\right)
\vert
\mathcal{F}_{n+1}^{23}\right],
\end{equation}%
\begin{equation}
\varsigma_{n}\left(  \varphi\right)  =\overline{E}\left[  \overline{Z}%
_{n}\varphi\left(  \theta_{n}\right)
\vert
\mathcal{F}_{n}^{23}\right]  =\overline{\zeta}_{n}\pi_{n}\left(
\varphi\right)
\end{equation}
and
\begin{equation}
P_{n}\left(  x,\varphi\right)  =\overline{E}\left[  \varphi\left(
\theta_{n+1}\right)
\vert
\theta_{n}=x\right].
\end{equation}

The main result of this subsection is the following

\begin{theorem}
$\forall n\in%
\mathbb{N}
$ , $\forall\varphi\in L^{\infty}\left(
\mathbb{R}
;%
\mathbb{R}
\right)  $,
\begin{equation}
\varsigma_{n+1}\left(  \varphi\right)  =\varsigma_{n}\left(  P_{n}\left(
.,\overline{\Lambda}_{n+1}\left(  .,x,y\right)  \varphi\right)  \right)
\vert
_{x=\Delta X_{n},y=\Delta Y_{n}} \label{ZakaiDiscreteEq}%
\end{equation}
and%
\begin{equation}
\zeta_{n+1}=\varsigma_{n}\left(  P_{n}\left(  .,\overline{\Lambda}%
_{n+1}\left(  .,x,y\right)  \right)  \right)
\vert
_{x=\Delta X_{n},y=\Delta Y_{n}} \label{ZakaiNormalizerDiscreteEq}%
\end{equation}
where $\zeta_{0}=1$, $\varsigma_{0}$ is assumed known and
\begin{equation}
P_{n}\left(  x,\varphi\right)  =E\left[  \varphi\left(  \frac{x+f_{1,n}\left(
\left(  1-\vartheta\right)  x\right)  }{1+\alpha_{n}w_{n}\vartheta}\Delta
\tau_{n}+\frac{g_{1,n}\left(  x\right)  \sqrt{\Delta\tau_{n}}}{1+\alpha
_{n}w_{n}\vartheta}\xi_{n}^{1}\right)  \right]  \text{.}%
\end{equation}

\end{theorem}

\begin{proof}%
\begin{align*}
\varsigma_{n+1}\left(  \varphi\right)   &  =\overline{E}\left[  \overline
{Z}_{n+1}\varphi\left(  \theta_{n+1}\right)
\vert
\mathcal{F}_{n+1}^{23}\right] \\
&  =\overline{E}\left[  \overline{Z}_{n}\overline{\Lambda}_{n+1}\left(
\theta_{n+1},\Delta X_{n},\Delta Y_{n}\right)  \varphi\left(  \theta
_{n+1}\right)
\vert
\mathcal{F}_{n+1}^{23}\right] \\
&  =\overline{E}\left[  \overline{Z}_{n}\overline{E}\left[  \overline{\Lambda
}_{n+1}\left(  \theta_{n+1},\Delta X_{n},\Delta Y_{n}\right)  \varphi\left(
\theta_{n+1}\right)
\vert
\theta_{n},\mathcal{F}_{n+1}^{23}\right]
\vert
\mathcal{F}_{n+1}^{23}\right] \\
&  =\overline{E}\left[  \overline{Z}_{n}P_{n}\left(  \theta_{n},\overline
{\Lambda}_{n+1}\left(  .,x,y\right)  \varphi\right)
\vert
_{x=\Delta X_{n},y=\Delta Y_{n}}%
\vert
\mathcal{F}_{n+1}^{23}\right] \\
&  =\overline{E}\left[  \overline{Z}_{n}P_{n}\left(  \theta_{n},\overline
{\Lambda}_{n+1}\left(  .,x,y\right)  \varphi\left(  .\right)  \right)
\vert
\mathcal{F}_{n}^{23}\right]
\vert
_{x=\Delta X_{n},y=\Delta Y_{n}}\\
&  =\varsigma_{n}\left(  P_{n}\left(  .,\overline{\Lambda}_{n+1}\left(
.,x,y\right)  \varphi\right)  \right)
\vert
_{x=\Delta X_{n},y=\Delta Y_{n}}%
\end{align*}
The equation $\left(  \ref{ZakaiNormalizerDiscreteEq}\right)  $ is obtained
when we apply simply the formula $\zeta_{n+1}=\varsigma_{n}\left(
\mathbf{1}\right)  $.
\begin{align*}
P_{n}\left(  x,\varphi\right)   &  =\overline{E}\left[  \varphi\left(
\theta_{n+1}\right)
\vert
\theta_{n}=x\right] \\
&  =E\left[  \varphi\left(  \theta_{n+1}\right)
\vert
\theta_{n}=x\right] \\
&  =E\left[  \varphi\left(  \frac{x+\alpha_{n}\left(  1-w_{n}\left(
1-\vartheta\right)  x\right)  }{1+\alpha_{n}w_{n}\vartheta}\Delta\tau
_{n}+\frac{g_{1,n}\left(  x\right)  \sqrt{\Delta\tau_{n}}}{1+\alpha_{n}%
w_{n}\vartheta}\xi_{n}^{1}\right)  \right]  \text{.}%
\end{align*}

\end{proof}

\section{Numerical illustrations of the time continuous filtering}\label{NumericalEstimations}

\qquad In this section, we carry out some simulations in order to have an idea
on the behaviour of the optimal filter we have theoreticall studied in
previous sections. We use the Theorem \ref{TheoremZakaiEqSDE} to solve the
SPDE $\left(  \ref{ZakaiEqSDEMeasure}\right)  $. For the reasons of stabilty,
memory space and simulation time, we take relatively big space stepsize
$\Delta x=0.1$ and time stepsize $\Delta t=10^{-3}$ for the resolution of the
equation $\left(  \ref{ZakaiEqSDEMeasure}\right)  $. We simply use the
well-known Euler's numerical scheme\footnote{See the reference \cite{kloeden}.}.
The parameters are taken following \cite{fotsa,fotsa3}. The control strategy $u$
is given for every time $t\geq0$ by
\begin{equation}
u(t)=\sin^{2}\left(  \omega_{1}\left(  t-\varphi_{1}\right)  ^{2}\right)
\exp\left(  -\omega_{2}\left(  t-\varphi_{2}\right)  ^{2}\right)  \text{.}%
\end{equation}
The functions $\alpha,\beta$ and $\gamma$ are taken with the following form
\begin{equation}
\alpha\left(  t\right)  =p_{1}\left(  t\right)  +b_{1}\left(  1-\cos\left(
c_{1}t\right)  \right)  \left(  t-d_{1}\right)  ^{2}\text{, }\forall t\in%
\mathbb{R}
_{+}\text{,}%
\end{equation}%
\begin{equation}
\beta\left(  t,x\right)  =b_{2}\left(  1-\cos\left(  c_{2}t\right)  \right)
\left(  t-d_{2}\right)  ^{2}p_{2}\left(  x\right)  \text{, }\forall\left(
t,x\right)  \in%
\mathbb{R}
_{+}\times%
\mathbb{R}
\text{,}%
\end{equation}
and
\begin{equation}
\gamma\left(  t,x_{1},x_{2},x_{3}\right)  =b_{3}\left(  1-\cos\left(
c_{3}t\right)  \right)  \left(  t-d_{3}\right)  ^{2}\left(  x_{1}-\kappa
x_{3}\right)  x_{2}\text{, }\forall\left(  t,x\right)  \in%
\mathbb{R}
_{+}\times%
\mathbb{R}
^{3}\text{.}%
\end{equation}
$p_{1}$ is a nonnegative function of the time $t$ and $p_{2}$ is a real
positive function of $x.$ $\forall i\in\left\{  1,2,3\right\}  ,$ $b_{i}$,
$c_{i}$, and $d_{i}$ are positive coefficients corresponding respectively to
the maximal amplitude, the pulsation and the global maximun of $\alpha,\beta$
and $\gamma$. $\kappa$ is a positive constant regulating the evolution of the
rot volume with respect to the inhibition rate. The terms $1-cos(\left(
c_{i}t\right)  $ represent the seasonality probably due to climatic and
environmental variations. Concerning the random parts of the equations, the
functions $\delta_{i}$ are assumed constant (the upper bound for instance) and
$\kappa_{1}\left(  x\right)  =x\left(  1-x\right)  $ if $x\in\left]
0,1\right[  $ and is null elsewhere. The initial conditions are taken such as
$\theta\left(  0\right)  \in\left\{  0.05,0.75\right\}  $, $v\left(  0\right)
\in\left\{  0.25,0.50\right\}  $  and $\rho\left(  0\right)  \in\left\{
0.25,0.75\right\}  $.

The following table gives the assumed parameters values.

\begin{table}[ptbh]
\centering\label{TableParamFilter1}%
\begin{tabular}
[c]{|c|c|c|c|c|c|}\hline
\textbf{Parameters} & \textbf{Values} & \textbf{Source} & \textbf{Parameters}
& \textbf{Values} & \textbf{Source}\\\hline
$b_{1}$ & $5\ln\left(  10\right)  $ & \cite{fotsa3} & $v_{\max}$ & $1$ &
\cite{fotsa3}\\
$b_{2}$ & $v_{\max}\ln\left(  10^{5}v_{\max}\left(  1-\varepsilon\eta^{\ast
}\right)  \right)  /2$ & \cite{fotsa3} & $\varepsilon$ & $10^{-4}$ &
\cite{fotsa3}\\
$b_{3}$ & $v_{\max}\ln(10^{5}v_{\max})$ & \cite{fotsa3} & $\sigma$ & $0.9$ &
\cite{fotsa,fotsa3}\\
$c_{i},$ $i=1,2,3$ & $10\pi$ & \cite{fotsa,fotsa3} & $\kappa$ & $1$ &
\cite{fotsa3}\\
$d_{i},$ $i=1,2,3$ & $7.5\times10^{-1}$ & \cite{fotsa,fotsa3} & $\Delta t$ &
$10^{-3}$ & Assumed\\
$\omega_{1}$ & $25\pi$ & \cite{fotsa3} & $\eta\left(  t\right)  $ & $1/\left(
1+\varepsilon\right)  $ & \cite{fotsa3}\\
$\omega_{2}$ & $10$ & \cite{fotsa3} & $p_{1}\left(  t\right)  $ & $0$ &
\cite{fotsa3}\\
$\varphi_{1}$ & $0.4$ & \cite{fotsa3} & $p_{2}\left(  x\right)  $ & $1$ &
Assumed\\
$\varphi_{2}$ & $0.6$ & \cite{fotsa3} & $\delta_{1}=\delta_{i},$ $i=2,3$ &
$10^{-2}$ & Assumed\\\hline
\end{tabular}
\caption{Simulation parameters for the filtering}%
\end{table}

We display two groups of figures. The first one represents the dynamics both
of the inhibition rate and filter corresponding to each values of $\rho\left(
0\right)  \leq\theta\left(  0\right)  $. The second group of figures shows
relative errors of the filter corresponding to each values of $\rho\left(
0\right)  \leq\theta\left(  0\right)  $.

\begin{figure}[ptbh]
\centering\label{figInhOF}%
\begin{tabular}
[c]{cc}%
\raisebox{-0cm}{\includegraphics[
natheight=12.171500cm,
natwidth=16.138599cm,
height=6.1549cm,
width=8.1385cm
]{./OF1EstimInhRat11.jpg}} & \raisebox{-0cm}{\includegraphics[
natheight=12.171500cm,
natwidth=16.138599cm,
height=6.1549cm,
width=8.1385cm
]{./OF1EstimInhRat12.jpg}}\\
$\theta\left(  0\right)  =v\left(  0\right)  =0.05$ & $\theta\left(  0\right)
=0.05$ and $v\left(  0\right)  =0.5$\\
\raisebox{-0cm}{\includegraphics[
natheight=12.171500cm,
natwidth=16.138599cm,
height=6.1549cm,
width=8.1385cm
]{./OF1EstimInhRat21.jpg}} & \raisebox{-0cm}{\includegraphics[
natheight=12.171500cm,
natwidth=16.138599cm,
height=6.1549cm,
width=8.1385cm
]{./OF1EstimInhRat22.jpg}}\\
$\theta\left(  0\right)  =0.75$ and $v\left(  0\right)  =0.05$ &
$\theta\left(  0\right)  =0.75$ and $v\left(  0\right)  =0.5$%
\end{tabular}
\caption{Inhibition rate and optimal filter}%
\end{figure}

\begin{figure}[ptbh]
\centering\label{figErrOF}%
\begin{tabular}
[c]{cc}%
\raisebox{-0cm}{\includegraphics[
natheight=12.171500cm,
natwidth=16.138599cm,
height=6.1527cm,
width=8.1385cm
]{./OF1ReAbsErr11.jpg}} & \raisebox{-0cm}{\includegraphics[
natheight=12.171500cm,
natwidth=16.138599cm,
height=6.1527cm,
width=8.1385cm
]{./OF1ReAbsErr12.jpg}}\\
$\theta\left(  0\right)  =v\left(  0\right)  =0.05$ & $\theta\left(  0\right)
=0.05$ and $v\left(  0\right)  =0.5$\\
\raisebox{-0cm}{\includegraphics[
natheight=12.171500cm,
natwidth=16.138599cm,
height=6.1527cm,
width=8.1385cm
]{./OF1ReAbsErr21.jpg}} & \raisebox{-0cm}{\includegraphics[
natheight=12.171500cm,
natwidth=16.138599cm,
height=6.1527cm,
width=8.1385cm
]{./OF1ReAbsErr22.jpg}}\\
$\theta\left(  0\right)  =0.75$ and $v\left(  0\right)  =0.05$ &
$\theta\left(  0\right)  =0.75$ and $v\left(  0\right)  =0.5$%
\end{tabular}
\caption{Relative absolute estimation error of the optimal filter}%
\end{figure}

\newpage

Looking at the simulations, the filter display fairly good behaviour. The estimation seems better when started soon, that is $\theta$, $v$, $\rho$ are relatively small. We also note that the variance of the absolute relative error is often big and we think it is due to the strong nonlinearity of the model, the small size of parameters $\delta_{i}$, $i=1,2,3$ and even the stepsizes of the
numerical scheme.

\section{Discussion}\label{DiscussionFiltering1}

\qquad This work is concerned by a filtering problem on anthracnose disease
dynamics. The aim has been was to provide an estimation of the inhibition rate based on
the assumption that the fruit volume and the rotted volume are easier to know.
Our used approch is similar with the one in the references \cite{fotsa3} except
that we assumed a noised dynamics. The noise has been modelled by Brownian motions
in order to keep a certain regularity on the solutions although taking into
account uncontrolled parameters variations (changes on at least climate and
environment) and errors on measurements. We have proposed and proved the
well-posedness for two modelling of the noised dynamics of the observations
trying to remain realistic. That work has been done both for a within host
version and a space distributed version. The first modelling seems to be more
natural but prensents some singularities in the noise. Those singularities
make difficult the application of classical filtering theory. We have then
proposed through a logistic transformation the second modelling which keeps roughly speaking the
same properties is easier to manage.

The filtering procedure has consisted into the determination of the law of the
inhibition rate at each time conditionally upon the fruit volume and the rotted
volume measurements. We have derived for that objective the Zakai and the
Kushner-Stratonovich equations respectively for the unnormalized and the
normalized conditional distributions. Unfortunately, we have restricted
ourselves to the non-spatial model because the spatial distributed model
requires more sophisticated techinical tools. Indeed, the problem consists in
that case to find a measure valued process operating on a functional space,
since at each fixed time the inhibition rate is not anymore a real but a
function of the space variable. However, we think that it might be possible to
deal with that problem if we consider gaussian spaces\footnote{See the Chapter
5 of the book \cite{malliavin}.} and existing works such as
\cite{boulanger,loges,bogachev,bogachev11,bogachev13} on resolution of
Fokker-Plank equations on infinite dimensional spaces. Additionaly to the main
filtering problem, we have also study related realistic problems such as
prevision and discrete filtering. That has appeared important to the authors since the observations
are often discrete and incomplete.

In order to illustrate numerically the filter behaviour, we have carried out
several simulations solving a stochastic partial differential equation corresponding to the
unnormalized conditional distribution. Following the literature
\cite{bobrovsky82,bobrovsky,katzur,steinberg,yaesh,zeitouni}, the filter is
more effective as the size of the noise is weaker. Unfortunatly, that induces an increase of the variance of the filter since there is a division by the variance of the observation noise. Moreover, it makes more difficult the computations in terms of stability of the numerical scheme, time and memory required. We suggest based on the theory of Luenberger-like observers (see \cite{luenberger}) to multiply the terms coming fromthe observations in filtering equations by an adequate constant. We could also
replace those terms by the minimum between them and an adequate constant. That changes may permit to reduce the variance of the filter and unfortunately could neglect the informations brought by the observations. We expect in future studies to survey
rigourously the properties of our filters since as far as we know that has been tried in very
restrictive cases in the literature (See for instance references
\cite{bobrovsky82,bobrovsky,katzur,steinberg,yaesh,zeitouni}).

\begin{acknowledgement}
The author thanks the International Center for Pure and Applied Mathematics
(ICPAM) and all its partners for have granted the first author with the Ibni
Oumar Mahamat Saleh Prize 2014 that has partially financed a research stay of
three months at the Mathematics Institute of Marseille (I2M). The basis of
this work has been essentially realized during that stay. The author also
thanks professor Etienne PARDOUX for its precious advices.
\end{acknowledgement}

\end{document}